%% file: main.tex
\RequirePackage{amsmath}
\documentclass{sig-alternate-05-2015}
\usepackage{ifthen}
\newcommand{\techRep}{true} 
\newcommand{\iftechrep}{\ifthenelse{\equal{\techRep}{true}}}
\usepackage{
amssymb,stmaryrd}
\usepackage{xcolor}

\usepackage{tikz}
\usetikzlibrary{arrows,calc,automata}
\tikzset{LMC style/.style={>=stealth',every edge/.append style={thick},every state/.style={minimum size=20,inner sep=0}}}
\newtheorem{theorem}{Theorem}
\newtheorem{corollary}[theorem]{Corollary}
\newtheorem{lemma}[theorem]{Lemma}
\newtheorem{proposition}[theorem]{Proposition}
\newtheorem{example}[theorem]{Example}

\newcommand{\A}{\mathcal{A}}%
\newcommand{\Bad}{\mathit{Bad}}%
\newcommand{\card}[1] {\vert #1 \vert}
\newcommand{\cd}{\mathit{cd}}
\newcommand{\Distr}[1]{\mathit{Distr}(#1)}
\newcommand{\dist}{\mathit{dist}}
\newcommand{\Class}{\mathit{Class}}%
\newcommand{\E}{\mathcal{E}}
\newcommand{\Good}{\mathit{Good}}%
\newcommand{\io}[1]{[#1]}%
\newcommand{\Lang}{\mathcal{L}}

\newcommand{\low}{\mathit{low}}
\newcommand{\lr}{\mathit{lr}}
\newcommand{\N}{\mathbb{N}}%
\renewcommand{\P}{\mathcal{P}}%
\newcommand{\pr}{\mathit{pr}}
\newcommand{\Q}{\mathbb{Q}}%
\newcommand{\sub}{\mathit{sub}}%
\newcommand{\supp}{\mathit{supp}}
\newcommand{\Test}{\mathit{Test}}
\newcommand{\Unreach}{\mathit{Unreach}}
\newcommand{\vv}[1]{\llbracket #1 \rrbracket}

\newcommand{\myparagraph}[1]{\smallskip\noindent\textbf{#1}}%

\newenvironment{qtheorem}[1]{%
{\smallskip\par\sc Theorem #1.}
\begin{itshape}%
}{%
\end{itshape}%
}
\newenvironment{qlemma}[1]{%
{\smallskip\par\sc Lemma #1.}%
\begin{itshape}%
}{%
\end{itshape}%
}

\newenvironment{qproposition}[1]{%
{\smallskip\par\sc Proposition #1.}
\begin{itshape}%
}{%
\end{itshape}%
}

\begin{document}

\CopyrightYear{2016} 
\setcopyright{acmlicensed}
\conferenceinfo{LICS '16,}{July 05 - 08, 2016, New York, NY, USA}
 \isbn{978-1-4503-4391-6/16/07}\acmPrice{\$15.00}
\doi{http://dx.doi.org/10.1145/2933575.2933608}

\title{Distinguishing Hidden Markov Chains%
\iftechrep{%
\titlenote{This is the full version of a LICS'16 paper.
}
}{%
\titlenote{A full version of this paper is available at \texttt{http://arxiv.org/abs/1507.02314}
}
}
}
\numberofauthors{2}

\author{
\alignauthor
Stefan Kiefer\\
\affaddr{University of Oxford, UK}
\alignauthor
A. Prasad Sistla\\
\affaddr{University of Illinois at Chicago, USA}
}

\maketitle
 
\input{abstract}




\input{intro}

\input{related}

\input{prelim}

\input{distinguish}

\input{monitors}

\input{profiles}

\input{verification}

\input{conclusions}

\myparagraph{Acknowledgments.}
Stefan Kiefer is supported by a University Research Fellowship of the
Royal Society. Prasad Sistla is partly supported by the  NSF
grants CCF-1319754 and CNS-1314485.

\bibliographystyle{abbrv}
\bibliography{db}

\iftechrep{
\clearpage
\appendix

\input{app-monitors}

\input{app-profiles}

\input{app-verification}

}{}

\end{document}

%% file: abstract.tex
\begin{abstract}
Hidden Markov Chains (HMCs) are commonly used mathematical models of probabilistic systems.
They are employed in various fields such as speech recognition, signal processing, and biological sequence analysis.
Motivated by  applications in stochastic runtime verification, we consider the problem of distinguishing two given HMCs based on a single observation sequence that one of the HMCs generates.
More precisely, given two HMCs and an observation sequence, a distinguishing algorithm is expected to identify the HMC that generates the observation sequence.
Two HMCs are called distinguishable if for every $\varepsilon > 0$ there is a distinguishing algorithm whose error probability is less than~$\varepsilon$.
We show that one can decide in polynomial time whether two HMCs are distinguishable.
Further, we present and analyze two distinguishing algorithms for distinguishable HMCs.
The first algorithm makes a decision after processing a fixed number of observations, and it exhibits two-sided error.
The second algorithm processes an unbounded number of observations, but the algorithm has only one-sided error. 
The error probability, for both algorithms, decays exponentially with the number of processed observations. 
We also provide an algorithm for distinguishing multiple HMCs. 
\end{abstract} 

%% file: intro.tex
\section{Introduction}
Hidden Markov Chains (HMCs) are commonly used mathematical models of probabilistic systems.
They are specified by a Markov Chain, capturing the probabilistic behavior of a system, and an observation function specifying the outputs generated from each of its states.
Figure~\ref{fig-fig} depicts two example HMCs $H_1, H_2$, with observations $a$ and~$b$.
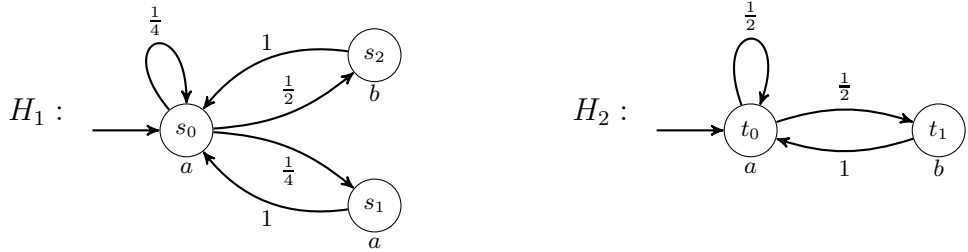
\begin{figure*}[ht]
\begin{center}
\begin{tikzpicture}[scale=2.5,LMC style]
\node (H1) at (-1.8,0.1) {\large $H_1 : $};
\node[state] (s0) at (-1,0) {$s_0$};
\node (a0) at (-1,-0.2) {$a$};
\node[state] (s1) at ( 0,-0.4) {$s_1$};
\node (b01) at (-0,-0.6) {$a$};
\node[state] (s2) at ( 0,0.4) {$s_2$};
\node (b0) at (-0,+0.2) {$b$};
\node (H2) at (1.2,0.1) {\large $H_2 : $};
\node[state] (t0) at (+2,0) {$t_0$};
\node (a1) at (+2,-0.2) {$a$};
\node[state] (t1) at (+3,0) {$t_1$};
\node (b1) at (+3,-0.2) {$b$};
\path[->] (-1.5,0) edge (s0);
\path[->] ( 1.5,0) edge (t0);
\path[->] (s0) edge [loop,out=130,in=90,looseness=13] node[above] {$\frac14$} (s0);
\path[->] (s0) edge [bend right=18] node[above] {$\frac12$} (s2);
\path[->] (s0) edge [bend left=18] node[below] {$\frac14$} (s1);
\path[->] (s2) edge [bend right=30] node[above] {$1$} (s0);
\path[->] (s1) edge [bend left=30] node[below] {$1$} (s0);
\path[->] (t0) edge [loop,out=110,in=70,looseness=13] node[above] {$\frac12$} (t0);
\path[->] (t0) edge [bend left=18] node[above] {$\frac12$} (t1);
\path[->] (t1) edge [bend left=18] node[below] {$1$} (t0);
\end{tikzpicture}
\end{center}
\caption{Two HMCs. Here $H_1$ and $H_2$ are distinguishable (see Example~\ref{ex-disting-freq}) and hence not equivalent.}
\label{fig-fig}
\end{figure*}
We consider finite-state HMCs in this paper.
An HMC randomly generates a (conceptually infinite) string of
observations. 
The states producing the observations are not observable
 (note that $s_0$~and~$s_1$ output the same observation~$a$ in the example).
This motivates the term \emph{hidden}.

HMCs are widely employed in fields such as speech recognition (see~\cite{Rabiner89} for a tutorial),
gesture recognition~\cite{Gesture},
musical score following~\cite{MusicalScore},
signal processing~\cite{SignalProcessing},
and climate modeling~\cite{Weather}.
HMCs are heavily used in computational biology~\cite{HMM-comp-biology},
more specifically in DNA modeling~\cite{DNA-modeling} and biological sequence analysis~\cite{durbin1998biological},
including protein structure prediction~\cite{ProteinStructure}, detecting similarities in genomes~\cite{Homology} and gene finding~\cite{GeneFinding}.
Following~\cite{LyngsoP02}, applications of HMCs are based on two basic problems, cf.~\cite[Chapter~2]{book-Fraser-HMM}:
The first one is, given an observation string and an HMC, what is the most likely sequence of states that produced the string?
This is useful for areas like speech recognition, see~\cite{Rabiner89} for efficient algorithms based on dynamic programming.
The second problem is, given an observation string and multiple HMCs, identify the HMC that is most likely to produce the observation.
This is used for classification.

The second problem raises a fundamental question, which we address in this work:
Given two HMCs, and assuming that one of them produces a random single observation sequence, is it even possible to identify the producing HMC with a high probability?
And if yes, how many observations in that observation sequence are needed?
At its heart, this question is about comparing two HMCs in terms of their (distributions on) observation sequences.
To make this more precise, let a \emph{monitor} for two given HMCs~$H_1, H_2$ be an algorithm that reads (increasing prefixes of) a single observation sequence, and at some point outputs ``$H_1$'' or ``$H_2$''.
The \emph{distinguishability} problem asks for two given HMCs $H_1, H_2$, whether for all $\varepsilon > 0$ there is a monitor such that for both $i=1,2$, if the monitor reads a random observation sequence produced by~$H_i$, then with probability at least $1-\varepsilon$ the monitor outputs~``$H_i$''.

A related problem is \emph{equivalence} of HMCs.
Two HMCs are called \emph{equivalent} if they produce the same (prefixes of) observation sequences with the same probability.
Equivalence of HMCs has been well-studied and can be decided in polynomial time, using algorithms based on linear algebra, see e.g.\ \cite{Ito-Equivalence,Tzeng,Doyen-Equivalence}.
The exact relation between equivalence and distinguishability depends on whether a monitor has access to a single random observation sequence or to multiple such sequences.

\begin{figure*}[ht]
\begin{center}
\begin{tikzpicture}[scale=2.5,LMC style]
\node (H1) at (-1.8,0.1) {\large $H_1 : $};
\node[state] (s0) at (-1,0) {$s_0$};
\node (a0) at (-1,-0.2) {$a$};
\node[state] (s1) at ( 0,0) {$s_1$};
\node (b0) at (-0,-0.2) {$b$};
\node (H2) at (1.2,0.1) {\large $H_2 : $};
\node[state] (t0) at (+2,0) {$t_0$};
\node (a1) at (+2,-0.2) {$a$};
\node[state] (t1) at (+3,0) {$t_1$};
\node (b1) at (+3,-0.2) {$b$};
\path[->] (-1.5,0) edge (s0);
\path[->] ( 1.5,0) edge (t0);
\path[->] (s0) edge [loop,out=110,in=70,looseness=13] node[above] {$\frac12$} (s0);
\path[->] (s0) edge node[above] {$\frac12$} (s1);
\path[->] (s1) edge [loop,out=110,in=70,looseness=13] node[above] {$1$} (s1);
\path[->] (t0) edge [loop,out=110,in=70,looseness=13] node[above] {$\frac23$} (t0);
\path[->] (t0) edge node[above] {$\frac13$} (t1);
\path[->] (t1) edge [loop,out=110,in=70,looseness=13] node[above] {$1$} (t1);
\end{tikzpicture}
\end{center}
\caption{Two HMCs. Here $H_1$ and~$H_2$ are not distinguishable but not equivalent.}
\label{fig-non-distinguishable}
\end{figure*}
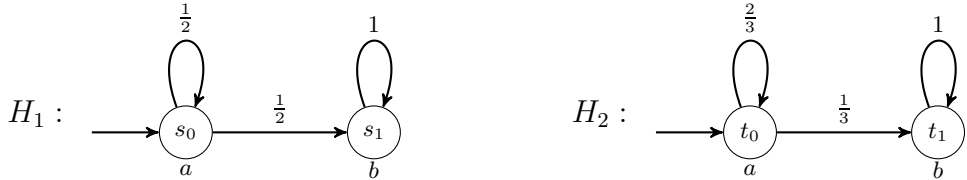

\begin{itemize}
\item[(1)]
Consider first a notion of a ``monitor'' that has access to \emph{several} random observation sequences,
each generated starting from the same initial state. 
Call this a \emph{multi-monitor}.
If the two given HMCs are equivalent then even a multi-monitor can only guess.
Now assume the two HMCs are not equivalent.
It is known (see e.g.~\cite{Tzeng}) that then there exists a linear-length prefix of the observation sequence that is more likely in one HMC than in the other HMC.
A multi-monitor could exploit the law of large numbers and only count how often that particular observation prefix occurs.
Hence for multi-monitors, distinguishability and non-equivalence coincide.
\item[(2)]
Consider now a monitor that has access to only a \emph{single} random observation sequence.
Here, non-equivalence does not imply distinguishability:
loosely speaking, for some HMCs it is the case that while the observation prefix is increasing, the evidence added by each new observation does not help the monitor enough to make up its mind about which HMC produces the sequence.
Figure~\ref{fig-non-distinguishable} shows an example of two HMCs that
are neither equivalent nor distinguishable. (On the
other hand, the HMCs in Figure~\ref{fig-fig} are not equivalent, but
are distinguishable as shown later in Section~\ref{sec-distinguish}).
\end{itemize}
We assume in the rest of the paper that a monitor has access to only a single random observation sequence.
This is the more natural version of the problem, both from the point of view of the motivation mentioned above and from our application in stochastic runtime monitoring.



We prove that the distinguishability problem is decidable in polynomial time.
We establish this result by showing that two HMCs are distinguishable if and only if their \emph{total variation distance} is equal to~$1$.
This distance measure for HMCs was studied in~\cite{14CK-LICS}, and a polynomial-time algorithm for deciding whether the distance of two HMCs is~$1$ was given there.
That polynomial-time algorithm includes a mechanism for checking whether two given HMCs are equivalent (but also needs other ingredients).

It is important to note that deciding distinguishability does not readily provide a family of monitors as required by the definition of distinguishability; it only guarantees their existence.
Developing a family of monitors (one for any desired error bound $\varepsilon > 0$) requires more insights.
Inspired by the area of \emph{sequential
  analysis}~\cite{wetherill}, we design  monitors that
track the \emph{likelihood ratio} of the sequence of
observations. 
However,  estimating the error probability  of the monitors is
challenging, since one needs a bound on the change of the likelihood
ratio per observation.
Unfortunately, such a bound does not exist for HMCs in general, not even on the difference of the \emph{log}-likelihood ratio (see Example~\ref{ex-unbounded}).
Hence, in this paper we take a different route: We consider a
different class
of monitors that  translate the given random observation sequence into a certain kind of non-homogenous ``random walk'' with \emph{bounded} step size.
This allows us to employ martingale techniques, specifically Azuma's inequality, to prove error bounds that decay exponentially with the number of observations the monitor makes.
Then we show that the error bounds from a random-walk monitor carry over to a likelihood-based monitor. 

More specifically, we present two likelihood-based monitors for distinguishable HMCs.
The first one makes a decision after reading a fixed number of observation symbols.
This number is chosen depending on the desired error bound:
we show that for an error probability~$\varepsilon$ it suffices to read the prefix of length $C  \log\frac{1}{\varepsilon}$, where $C>0$ is a polynomial-time computable constant.
This error is two-sided, i.e., the monitor may mistake $H_1$ for~$H_2$ and vice versa.

The second monitor has only one-sided error: observation sequences from~$H_1$ are almost always (i.e., with probability~$1$) recognized as stemming from~$H_1$.
However, on sequences generated by~$H_2$, with high probability the monitor never gives an answer. 
This is useful in applications such as runtime verification (see Section~\ref{sec-verification}).
The expected number of observations from~$H_1$ that the monitor
processes before giving its decision is
$O(\log\frac{1}{\varepsilon})$, while ensuring an error probability of
at most~$\varepsilon$ on observations from~$H_2$. For this class of
monitors, we have a polynomial-time algorithm that computes an $O(\log\frac{1}{\varepsilon})$ upper
bound on the expected number of observations from~$H_1$ before a decision
is given. 


\myparagraph{Main Contributions.} 
\begin{itemize}
\item We show that the distinguishability problem can be decided in polynomial time (Section~\ref{sec-distinguish}).
\item We design two classes of likelihood-based monitors that accomplish the following tasks ($\varepsilon>0$ is an error bound):
\begin{enumerate}
\item[(1)] After $O(\log \frac{1}{\varepsilon})$ observations (the exact number can be efficiently computed from the given HMCs) the first monitor class provides a guess about the source of the observations, such that the probability that the guess is wrong is at most~$\varepsilon$ (Section~\ref{sub-two-sided}).
    This can be extended to more than two HMCs (Section~\ref{sub-multiple}).
\item[(2)] For the second monitor class, if $H_1$ produces the
  observation sequence then the monitor raises an alarm almost surely,
  and after an \emph{expected} number of $O(\log
  \frac{1}{\varepsilon})$ observations (such an upper bound can be
  efficiently computed from the given HMCs and $\varepsilon$);
    if $H_2$ produces the observation sequence then, with probability at least $1-\varepsilon$, the monitor never raises an alarm (Section~\ref{sub-monitor-one-sided}).
\end{enumerate} 
\item
We apply our results to stochastic runtime verification, where a monitor should distinguish correct and faulty behaviour of a \emph{single} stochastic system.
This yields polynomial-time decidability of monitorability as defined in~\cite{Sistla11}, as well as efficient runtime monitors for stochastic systems, see Section~\ref{sec-verification}.
\end{itemize}
Missing proofs can be found \iftechrep{in the appendix}{in~\cite{16KS-TR}}.

%% file: related.tex
\myparagraph{Related Work.}
The area of \emph{sequential analysis} in statistics, pioneered by Wald (see \cite{wetherill}), deals with the problem of hypothesis testing using repeated and unbounded sampling.
A line of work going back to Phatarfod \cite{Phatarfod,Swamy,Schmitz} investigated the application of sequential analysis, more specifically the sequential probability ratio test, to Markov chains.
Similar to our work, the goal in the above works is to identify a Markov chain among several, in this case using likelihood ratios.
A monitor algorithm is derived by keeping track of likelihood ratios:
it gives notice once the likelihood ratio drops below or exceeds some fixed threshold.
One problem with this approach is that error probabilities can only be estimated---not bounded---by the heuristic assumption that the excess over the threshold is not big.
This assumption is not always true.
A more important difference from our work is that the observation in each state equals the state,
in other words, the Markov chains are not hidden.

There is early related work that is more specific to HMCs.
The paper~\cite{probDistMeasure} aims at \emph{measuring} a certain distance between two HMCs by running one of them.
This is in spirit close to our work, as a positive distance in their sense could be transformed to a monitor.
However, the authors place strong assumptions on the Markov chains, in particular ergodicity.
If this assumption is removed, their distance can be different for different runs,
 and the existence of a lower bound on the possible distances is unclear.


Work by Alur et al.~\cite{AlurCY95} also aims at distinguishing probabilistic models, but there are important differences.
First, they consider Markov Decision Processes rather than Markov chains, i.e., they consider \emph{strategies} to distinguish two such processes, which is a more general, and computationally harder problem (they show PSPACE- and EXPTIME-completeness results).
Second, their problems are defined such that the exact values of the transition probabilities is unimportant.
In our case this is different.

The work in~\cite{LyngsoP02} deals with comparing two HMCs in terms of various distance measures.
Among other results, they show NP-hardness of computing and approximating the $\ell_1$-distance.
The HMCs considered there generate distributions on \emph{finite} strings of observations, as each HMC has a dedicated end state, reached with probability~$1$, where the HMC ``stops''.
Such HMCs form a subclass of HMCs, whereas we consider general HMCs.

Our work on distinguishability is inspired by the work on monitorability that was defined in~\cite{Sistla11}.
In \cite[Section 4.1]{Sistla11} a notion of \emph{strong monitorability} is proposed and it is shown that deciding it is PSPACE-complete.
By our results in Section~\ref{sec-verification}, strong monitorability corresponds to a stronger form of distinguishability, so the latter is PSPACE-complete as well.
In light of this it might be surprising that (general) distinguishability turns out to be decidable in polynomial time. 
In~\cite{Sistla11} it was wrongly claimed that \emph{monitorability} is undecidable for finite-state systems.
Our result not only shows that it is decidable, but also gives a polynomial-time decision procedure.

Our work on exponentially decaying monitors is inspired by the exponentially converging monitorable systems defined in \cite{Sistla14}.
The algorithms presented there are for a very restricted class of HMCs, whereas our monitors work for all pairs of distinguishable HMCs. 

Closely related to some of our results is a very recent work by Bertrand et.\ al.~\cite{BertrandHL16}.
This paper also exploits the results of~\cite{14CK-LICS} to obtain polynomial-time decidability of ``AA-diagnosability'' of stochastic systems, a problem related to monitorability (Section~\ref{sec-verification}).
Although the technical report of our work had been available~\cite{16KS-TR},
the results in~\cite{BertrandHL16} were obtained independently and are largely orthogonal to ours:
whereas we focus on constructing specific monitors with computable error bounds,
they investigate the decidability and complexity of several variants of diagnosability.


%% file: prelim.tex
\section{Definitions}
\label{sec-prelim}
\myparagraph{Notation.}
For a countable set $S$, a probability distribution $\psi$ over~$S$
is a function $\psi: S\to [0,1]$ such that $\sum_{s\in
  S}\psi(s)\:=1$. For an element $s\in S$, we let $\delta_s$ denote the unique
distribution with $\delta_s(s)=1$. We let $\Distr{S}$ denote the set
of all distributions over~$S$. We let $S^*,S^{\omega}$ respectively  denote the set
of finite sequences (\emph{strings}) and the set of infinite sequences of symbols from~$S$. If $S$ is a finite set then we let $\card{S}$ denote its
cardinality. For any $u\in S^*$, we let $\card{u}$ denote its
length. 
For any real number $x$, we let $\card{x}$ denote its
absolute value.

\myparagraph{Hidden Markov Chains.} A Markov chain is a triple $G = (S, R, \phi)$ where $S$ is a set of states, $R\subseteq S\times S$, and $\phi:R\to (0,1]$ is such that $\sum_{t : (s,t) \in R} \phi(s,t)\:=1$ for all $s\in S$. A Markov chain $G$
 and an initial state $s \in S$
 induce a probability measure, denoted by~$\P_s$, on measurable
 subsets of $\{s\}S^\omega$ in the usual way:
more precisely, we consider the $\sigma$-algebra generated by the cylinder sets $\{s_0 s_1 \cdots s_n\} S^\omega$ for $n \ge 0$ and $s_0 = s$ and $s_i \in S$, with the probability measure~$\P_s$ such that
\[
 \P_s(\{s_0 s_1 \cdots s_n\} S^\omega) = \prod_{i=1}^{n} \phi(s_{i-1}, s_i)\,.
\]
 Let $\Sigma$ be a finite set. A \emph{Hidden Markov Chain (HMC)}, with
observation alphabet $\Sigma$, is a triple $(G, O, s_0)$, 
where $G = (S, R, \phi)$ is a Markov chain,
and $O: S \to \Sigma$ is the observation function, 
and $s_0 \in S$ is the initial state.
We may write $\P$ for~$\P_{s_0}$.
For $\Lang \subseteq \Sigma^\omega$ we define
the inverse observation function
\[
 \io{\Lang} := \{s_0 s_1 \cdots \in S^\omega \mid O(s_0) O(s_1) \cdots \in \Lang\}\,.
\]

\myparagraph{Monitors.}
A \emph{monitor} $M : \Sigma^* \to \{\bot,1\}$ is a computable function with the property that, for any $u \in \Sigma^*$, if $M(u) = 1$ then $M(u v) = 1$ for every $v \in \Sigma^*$.
Let $\Lang(M) \subseteq \Sigma^\omega$ denote the set of infinite
sequences that have a prefix~$u$ with $M(u)=1$. (Intuitively, $\Lang(M)$ is
the set of observation sequences which the monitor decides to have been
generated by the first HMC among a pair of such HMCs.)
Given an HMC, the event $\io{\Lang(M)}$ is measurable, as it is a countable union of cylinder sets.

\myparagraph{Distinguishability.}
Given two HMCs $H_1, H_2$ with the same observation alphabet~$\Sigma$,
we write $\P_1, \P_2, \io{\cdot}_1, \io{\cdot}_2$ for their associated probability measures and inverse observation functions.
HMCs $H_1, H_2$ are called \emph{distinguishable} 
 if for every $\varepsilon > 0$ there exists a monitor~$M$ such that
 \[
  \P_1(\io{\Lang(M)}_1) \ \ge \ 1-\varepsilon \quad \text{and} \quad 
  \P_2(\io{\Lang(M)}_2) \ \le \ \varepsilon\,. 
 \]

%% file: distinguish.tex
\section{Polynomial-Time Decidability of the Distinguishability Problem}
\label{sec-distinguish}
For two HMCs $H_1, H_2$
define the \emph{(total variation) distance} between $H_1$ and~$H_2$,
denoted by $d(H_1,H_2)$, as follows:
\[
 d(H_1, H_2) := \sup_{E \subseteq \Sigma^\omega} \left| \P_1(\io{E}_1) - \P_2(\io{E}_2) \right| \;,
\]
where the supremum ranges over all \emph{measurable} subsets of~$\Sigma^\omega$.
It is shown in~\cite{14CK-LICS} that the supremum is in fact a maximum.
In particular, if $d(H_1, H_2) = 1$ then there exists a measurable set $E \subseteq \Sigma^\omega$
 with $\P_1(\io{E}_1) = 1$ and $\P_2(\io{E}_2) = 0$.
We show:
\newcommand{\stmtpropreductiontodistone}{
HMCs $H_1, H_2$ are distinguishable if and only if $d(H_1, H_2) = 1$.
}
\begin{proposition} \label{prop-reduction-to-dist-one}
\stmtpropreductiontodistone
\end{proposition}
\begin{proof}
Let $H_1, H_2$ be two given HMCs.
We show that $H_1, H_2$ are distinguishable if and only if $d(H_1, H_2) = 1$.
\begin{itemize}
\item \emph{``if''}:
Let $d(H_1, H_2) = 1$.
Choose $\varepsilon > 0$ arbitrarily.
It follows from \cite[Theorem~7]{14CK-LICS} and the discussion after \cite[Proposition~5]{14CK-LICS} that
there are $k \in \N$ and $W \subseteq \Sigma^k$ such that
\begin{equation*} 
 \P_1(\io{W \Sigma^\omega}_1) \ge 1 - \varepsilon \quad \text{and} \quad 
 \P_2(\io{W \Sigma^\omega}_2) \le \varepsilon \,.
\end{equation*}
Construct a monitor~$M$ that outputs~$1$ after having read a string in~$W$.
Then we have $\Lang(M) = W \Sigma^\omega$.
It follows:
 \[
  \P_1(\io{\Lang(M)}_1) \ge 1-\varepsilon \quad \text{and} \quad
  \P_2(\io{\Lang(M)}_2) \le \varepsilon\,.
 \]
Since $\varepsilon$ was chosen arbitrarily,
 the HMCs $H_1, H_2$ are distinguishable.
\item \emph{``only if''}:
Let $H_1, H_2$ be distinguishable, i.e., 
for every $\varepsilon > 0$ there exists a monitor~$M_\varepsilon$
 such that
 \[
  \P_1(\io{\Lang(M_\varepsilon)}_1) \ge 1-\varepsilon \quad \text{and} \quad
  \P_2(\io{\Lang(M_\varepsilon)}_2) \le \varepsilon\,.
 \]
Then we have:
\begin{align*}
 d(H_1, H_2) 
 & = \sup_{E \subseteq \Sigma^\omega} \left| \P(\io{E}_1) - \P(\io{E}_2) \right| \\
 & \ge \sup_{\varepsilon > 0} \big( \P(\io{\Lang(M_\varepsilon)}_1) - \P(\io{\Lang(M_\varepsilon)}_2) \big) \\
 & \ge \sup_{\varepsilon > 0} \left( 1 - 2 \varepsilon \right) 
 \ = \ 1 
\end{align*}
\end{itemize}
This concludes the proof.
\end{proof}

It follows that HMCs $H_1, H_2$ are distinguishable iff there is a \emph{distinguishing event}, i.e., a set $E \subseteq \Sigma^\omega$ with $\P_1(\io{E}_1) = 1$ and $\P_2(\io{E}_2) = 0$.
\begin{example} \label{ex-disting-freq}
Consider the HMCs $H_1, H_2$ from Figure~\ref{fig-fig}.
By computing the stationary distributions, one can show that, the
distinguishing event $E$ is given by
\[
 E = \{\sigma_1 \sigma_2 \cdots \in \Sigma^\omega \mid \lim_{n \to \infty} \frac{f(n)}{n} = 5/7 \}\;,
\]
where $f(n)$ denotes the number of occurrences of~$a$ in the prefix $\sigma_1 \sigma_2 \cdots \sigma_n$,
is a distinguishing event for $H_1, H_2$.
Hence $H_1, H_2$ are distinguishable.
Here, counting the frequencies of the observations symbols suffices for distinguishing two distinguishable HMCs.
In general, this is not true: the \emph{order} of observations may matter.
\qed
\end{example}
Proposition~\ref{prop-reduction-to-dist-one} implies the following theorem:
\begin{theorem} \label{thm-qual-P}
One can decide in polynomial time whether given HMCs $H_1, H_2$ are distinguishable. 
\end{theorem}
\begin{proof}
In~\cite[Algorithm~1 and Theorem~21]{14CK-LICS} it is shown that, given two HMCs $H_1, H_2$, one can decide in polynomial time whether $d(H_1, H_2) = 1$.
(The algorithm given there solves $n_1$ linear programs, each with $n_1 + n_2$ variables, where $n_1, n_2$ is the number of states in~$H_1, H_2$, respectively.)
Then the result follows from Proposition~\ref{prop-reduction-to-dist-one}.
\end{proof}
Distinguishing events cannot in general be defined by monitors, as a monitor can reject an observation sequence only on the basis of a finite prefix.
Moreover, the decision algorithm for Theorem~\ref{thm-qual-P} can assure the \emph{existence} of a monitor for two given HMCs, but the decision algorithm does not provide useful monitors.
That is the subject of the next section.


%% file: monitors.tex
\section{Monitors}
\label{sec-monitors}
In this section, we present concrete monitors, with error bounds.
To this end we give some additional definitions in Section~\ref{sub-keeping-track},
where we also explain how monitors can keep track of certain conditional distributions.
We also introduce ``profiles'', a key concept for our proofs of error bounds.
In Sections \ref{sub-two-sided} and~\ref{sub-monitor-one-sided} we present monitors for distinguishable HMCs with two-sided and one-sided error, respectively.
In Section~\ref{sub-multiple} we provide a monitor for distinguishing among multiple HMCs.

For $i=1,2$, let $H_i= (G_i,O_i, s_{i,0})$ be two 
HMCs with the same observation alphabet~$\Sigma$, where $G_i = (S_i,R_i,\phi_i)$.
Without loss of generality we assume $S_1 \cap S_2 = \emptyset$.
Let $m := \card{S_1}+\card{S_2}$. 
We fix $H_1,H_2$ and $m$ throughout the section.

\subsection{Keeping Track of Probabilities and Profiles} \label{sub-keeping-track}

Let $i \in \{1,2\}$ and $\psi \in \Distr{S_i}$.
For $u \in \Sigma^*$ define 
\[
 \pr_i(\psi,u) := \sum_{s\in S_i}\psi(s)\cdot\P_{i,s}(\io{u\Sigma^{\omega}}_i)\,.
\]
Intuitively, $\pr_i(\psi,u)$ is the probability that the string $u$ is output by HMC~$H_i$ starting from
the initial distribution $\psi$.
For $W \subseteq \Sigma^m$ we also define
\[
 \pr_i(\psi, W) := \sum_{u \in W} \pr_i(\psi,u)\,,
\]
which is the probability that $H_i$ outputs a string in~$W$ starting from distribution~$\psi$.
For $u \in \Sigma \Sigma^*$ and $s,t \in S_i$ define
\[ 
 \sub_{i}(s,u,t) := \P_{i,s}(\io{u \Sigma^{\omega}}_i\cap S_i^{|u|-1} \{t\} S_i^{\omega})\,.
\]
Intuitively, $\sub_i(s,u,t)$ is the probability that $H_i$ outputs~$u$ and is then in state~$t$, starting from state~$s$.
We have:
\begin{equation} \label{eq-sub-to-pr}
 \pr_i(\psi, u) = \sum_{s \in S_i} \psi(s) \cdot \sum_{t \in S_i} \sub_{i}(s,u,t)
\end{equation}
For any $s, r \in S_i$ and $u \in \Sigma\Sigma^*$ and $a \in \Sigma$ we have:
\begin{equation*} 
 \sub_i(s, u a, r) = 
 \begin{cases}
 \sum_{t \in S_i} \sub_i(s,u,t) \phi_i(t,r) & \text{if $O_i(r) = a$} \\
 0 & \text{otherwise}
 \end{cases}
\end{equation*}
So if a monitor has kept track of the values $\sub_i(s,u,t)_{s,t \in S_i}$ for a prefix~$u$ of an observation sequence, it can, upon reading the next observation~$a$, efficiently compute $\sub_i(s,u a,t)_{s,t \in S_i}$ and, by~\eqref{eq-sub-to-pr}, also $\pr_i(\delta_{s_{i,0}}, u a)$.

For $u \in \Sigma^*$ define the \emph{likelihood ratio}
\[
 \lr(u) := \frac{\pr_2(\delta_{s_{2,0}}, u)}{\pr_1(\delta_{s_{1,0}}, u)}\;.
\]
Finally, for $u \in \Sigma \Sigma^*$ with $\pr_i(\psi,u)>0$, define the distribution $\cd_i(\psi,u)$ (which stands for ``conditional distribution'') as follows:
\[
 \cd_i(\psi,u)(t):= \frac{1}{\pr_i(\psi,u)}\cdot \sum_{s\in S_i} \psi(s) \cdot \sub_{i}(s,u,t)\quad \text{ for } t \in S_i
\]
Intuitively, $\cd_i(\psi,u)(t)$  is the conditional probability that $H_i$ is in state~$t$ given that it has output $u$ and started from~$\psi$.
As explained above, a monitor can efficiently keep track of $\lr(u)$ and~$\cd_i(\psi,u)$.

We say that a pair of distributions $(\psi_1, \psi_2) \in \Distr{S_1} \times \Distr{S_2}$ is \emph{reachable} in $(H_1, H_2)$ if there is $u \in \Sigma\Sigma^*$ with $\psi_i=\cd_i(\delta_{s_{i,0}},u)$  for $i=1,2$.
A \emph{profile} for $H_1, H_2$ is a pair $(\A, c)$ such that 
$\A : \Distr{S_1} \times \Distr{S_2} \to 2^{\Sigma^m}$ and
$c \in (0,1]$ and
\[
 \pr_1(\psi_1, \A(\psi_1, \psi_2)) - \pr_2(\psi_2, \A(\psi_1, \psi_2)) \ge c
\]
holds for all reachable pairs $(\psi_1, \psi_2)$ of distributions.
For the monitors presented in this section the following proposition is crucial.
\begin{proposition} \label{prop-profile}
Let HMCs $H_1, H_2$ be distinguishable.
Then there is a number~$c>0$, computable in time polynomial in
the sizes of $H_1,H_2$, such that there is a profile $(\A,c)$.
\end{proposition}

\subsection{Monitors with Two-Sided Error} \label{sub-two-sided}

In this and the next subsection, we assume that $H_1,H_2$ are distinguishable,
 and fix a profile $(\A,c)$. 
The monitors of this subsection take an observation sequence as input, and at some point output a value from $\{1,2,3\}$ indicating a decision regarding which of the two HMCs generated the observations.
An output of~$3$ indicates that neither of the HMCs could have generated it.
The monitors of this subsection have two-sided errors: the answers $1$ or~$2$ may be wrong (with a small probability).

We define a likelihood-based monitor~$M_2$ (the subscript denotes two-sided error) as follows. 
Monitor~$M_2$ runs in \emph{phases}; in each phase, the monitor receives $m$~observations. 
The monitor runs at most~$N$ phases, where $N \in \N$ is a parameter fixed in advance: choosing a larger~$N$ leads to smaller error probabilities.  
After reading an observation sequence $u$ of length~$N \cdot m$, it
computes the likelihood ratio $\lr(u)$.  
Monitor~$M_2$ outputs~$1$ if $\lr(u) < 1$, and~$2$ if $\lr(u) > 1$.
It may output either $1$ or~$2$ if $\lr(u) = 1$. 
Monitor~$M_2$ needs no access to the function~$\A$.

The following theorem says that the observation sequences for which
monitor~$M_2$ outputs~$1$ are much more likely to be  generated by~$H_1$.
By symmetry, the observation sequences for which~$M_2$ outputs~$2$ 
are much more likely to be  generated by~$H_2$. 
\newcommand{\stmtcoralgonelikeli}{
Consider the monitor~$M_2$ that reads the first $N \cdot m$ observations.
Let $\Lang(M_2) \subseteq \Sigma^\omega$ be the set of observation sequences for which $M_2$ outputs~$1$.
Then we have
\[
 \P_1(\io{\Lang(M_2)}_1) - \P_2(\io{\Lang(M_2)}_2) \ge 1 - 2 \exp\left(-\frac{c^2}{18}\cdot N\right)\,.
\]
Hence,
 \begin{align*}
  \P_1(\io{\Lang(M_2)}_1) \ &\ge \ 1 - 2 \exp\left(-\frac{c^2}{18}\cdot N\right) \quad \text{and} \\
  \P_2(\io{\Lang(M_2)}_2) \ &\le \ 2 \exp\left(-\frac{c^2}{18}\cdot N\right)\,. 
 \end{align*}
}
\begin{theorem} \label{thm-alg1-likeli}
\stmtcoralgonelikeli
\end{theorem}

Proving the bounds of Theorem~\ref{thm-alg1-likeli} is challenging 
due to the following reasons.
For $k \ge 0$ define a random variable $L_k : \{s_{1,0}\} S_1^\omega \to \Q$ by
\[
 L_k(s_{1,0} s_1 s_2 \cdots ) := \lr\big(O(s_{1,0}) O(s_1) O(s_2) \cdots O(s_{k-1})\big)\,.
\]
Denote by $\E_1$ the expectation with respect to~$\P_1$.
It was proved in \cite[proof of Proposition~6]{14CK-LICS} that 
$
\E_1(L_{k+1} \mid L_k = x) \ = \ x
$
holds for all $x \in \Q$, i.e., the sequence $L_0, L_1, \ldots$ is a martingale.
Unfortunately, the differences $\card{L_{k+1} - L_k}$ are not bounded,
neither are the differences $\card{\log L_{k+1} - \log L_k}$, as the following example shows.

\begin{figure*}[ht]
\begin{center}
\begin{tikzpicture}[scale=2.5,LMC style]
\node (H1) at (-1.8,0.1) {\large $H_1 : $};
\node[state] (s0) at (-1,0) {$s_0$};
\node (a0) at (-1,-0.2) {$a$};
\node[state] (s1) at ( 0,-0.4) {$s_1$};
\node (b01) at (-0,-0.6) {$a$};
\node[state] (s2) at ( 0,0.4) {$s_2$};
\node (b0) at (-0,+0.2) {$b$};
\node (H2) at (1.2,0.1) {\large $H_2 : $};
\node[state] (t0) at (+2,0) {$t_0$};
\node (a1) at (+2,-0.2) {$a$};
\node[state] (t1) at (+3,0) {$t_1$};
\node (b1) at (+3,-0.2) {$b$};
\path[->] (-1.5,0) edge (s0);
\path[->] ( 1.5,0) edge (t0);
\path[->] (s0) edge [loop,out=130,in=90,looseness=13] node[above] {$\frac13$} (s0);
\path[->] (s0) edge [bend right=18] node[above] {$\frac13$} (s2);
\path[->] (s0) edge  node[below] {$\frac13$} (s1);
\path[->] (s2) edge [bend right=30] node[above] {$\frac12$} (s0);
\path[->] (s1) edge [loop,out=20,in=340,looseness=13] node[right] {$1$} (s1);
\path[->] (s2) edge [loop,out=20,in=340,looseness=13] node[right] {$\frac12$} (s2);
\path[->] (t0) edge [loop,out=110,in=70,looseness=13] node[above] {$\frac12$} (t0);
\path[->] (t1) edge [loop,out=110,in=70,looseness=13] node[above] {$\frac12$} (t1);
\path[->] (t0) edge [bend left=18] node[above] {$\frac12$} (t1);
\path[->] (t1) edge [bend left=18] node[below] {$\frac12$} (t0);
\end{tikzpicture}
\end{center}
\caption{Two HMCs where the difference in log-likelihood ratios is unbounded}
\label{fig-unbounded}
\end{figure*}
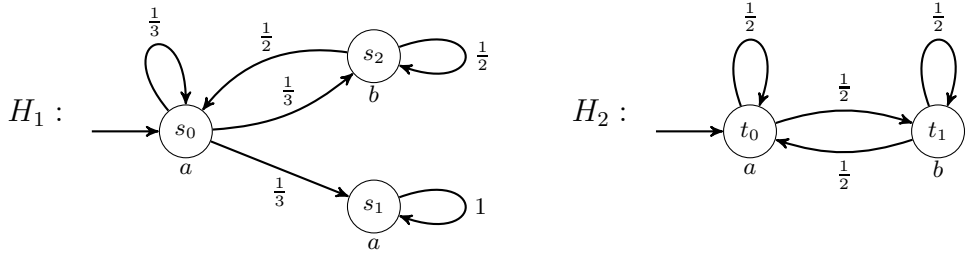

\begin{example} \label{ex-unbounded}
Consider the HMCs $H_1, H_2$ in Figure~\ref{fig-unbounded}.
For $n>1$, the probability that $H_1$ generates the
string $a^n$  is $(\frac{1}{3})^{n-1}+\frac{1}{3}\cdot \sum_{i=0}^{
  n-2}(\frac{1}{3})^i$ which is easily shown to be
$\frac{1}{2}(1+(\frac{1}{3})^{n-1})$, and the probability that 
$H_1$ generates $a^nb$ is $(\frac{1}{3})^n$. The corresponding
probabilities for $H_2$ are $(\frac{1}{2})^{n-1}$ and
$(\frac{1}{2})^n$, respectively.
Now consider any $\alpha\in \{s_0^ns_2\} \{s_0,s_1,s_2\}^\omega$, for some $n>1$. The two likelihood
ratios $L_n(\alpha)$ and $L_{n+1}(\alpha)$
corresponding to the length
$n$ and length $n+1$ prefixes of $\alpha$, are given by 
 $L_n(\alpha)=
 \frac{(\frac{1}{2})^{n-1}}{\frac{1}{2}(1+(\frac{1}{3})^{n-1})}$ and
$L_{n+1}(\alpha) = (\frac{3}{2})^n$. Since $n>1$, we see that $L_n(\alpha)<2\cdot
(\frac{1}{2})^{n-1} \le 1$. Hence,
$\frac{L_{n+1}(\alpha)}{L_{n}(\alpha)}> (\frac{3}{2})^n$.
So we have that $\log (L_{n+1}(\alpha))-\log(L_n(\alpha))>
n\cdot \log(\frac{3}{2})$,  which is
unbounded with increasing $n$. In a more general case, if $\alpha$ has $b$ appearing
infinitely often with an increasing number of $a$-symbols between
two successive $b$-symbols, then the difference in the log-likelihood ratio of
two successive prefixes of $\alpha$, with the second prefix ending
with $b$, is unbounded.
\qed
\end{example}

This problem of unbounded differences between subsequent log-likelihood ratios prohibits a standard error analysis of hypothesis-testing methods from sequential analysis~\cite{wetherill}.
Moreover, Azuma's inequality then does not yield an exponentially
decaying error bound. As a
consequence, we cannot directly prove the bounds of Theorem~\ref{thm-alg1-likeli}.
Therefore, in this subsection, we take a detour.
First we develop another  monitor~$M_2'$ that is not based on
likelihoods but is based on a random walk.
Then we prove error bounds for~$M_2'$.
Then we show that the error bounds for~$M_2'$  carry over to the
likelihood-based monitor~$M_2$. 

The monitor~$M_2'$ also runs in  $N$ phases, receiving
$m$~observations in each phase.  
The monitor maintains two probability distributions $\psi_1\in \Distr{S_1},\:\psi_2\in\Distr{S_2}$, and a variable~$x$ that takes rational values.
Initially, $\psi_1,\psi_2$ are set to $\delta_{s_{1,0}},\delta_{s_{2,0}}$ respectively, 
and $x$ is initialized to $0$.
The monitor keeps track of
$\psi_i\:=\cd_i(\delta_{s_{i,0}},u)$, for $i=1,2$, where $u$ is the observation string received thus far.
The variable~$x$ indicates a current estimate about which of the two HMCs is being observed:
 a negative value of~$x$ indicates a preference for~$H_1$; a positive value indicates a preference for~$H_2$.
In each phase, $M_2'$ waits until it gets the next $m$ observations and then updates $x$, $\psi_1$ and~$\psi_2$. 

We describe a phase of~$M_2'$.
Let $\psi_1, \psi_2, x$ be the values at the end of the previous phase.
Let 
$p_1 = \pr_1(\psi_1,\A(\psi_1,\psi_2))$ and 
$p_2 = \pr_2(\psi_2,\A(\psi_1,\psi_2))$.
By the definition of a profile we have $p_1 - p_2 \ge c > 0$.
Denote by $v \in \Sigma^m$ the string of observations received in the current phase.
Assume that $\pr_1(\psi_1,v)>0$ and $\pr_2(\psi_2,v)>0$ (i.e., $v$ can be generated with non-zero
probability by both $H_1,H_2$ from $\psi_1,\psi_2$ respectively). 
If $p_1+p_2 \le 1$ then $x$ is updated as follows:
\[
x := \begin{cases}
x - 1                         & \text{ if $v \in \A(\psi_1,\psi_2)$} \\
x + \frac{p_1+p_2}{2-p_1-p_2} & \text{ if $v \not\in \A(\psi_1,\psi_2)$}
\end{cases}
\]
If $p_1+p_2 > 1$ then $x$ is updated as follows:
\[
x := \begin{cases}
x - \frac{2-p_1-p_2}{p_1+p_2} & \text{ if $v \in \A(\psi_1,\psi_2)$} \\
x + 1                         & \text{ if $v \not\in \A(\psi_1,\psi_2)$}
\end{cases}
\]
Note that in all cases, the value of $x$ is increased or decreased by at most $1$.
After this, $\psi_1,\psi_2$ are set to $\cd_1(\psi_1,v)$ and $\cd_2(\psi_2,v)$ respectively, and the phase is finished.
On the other hand, if $\pr_1(\psi_1,v)>0$ and $\pr_2(\psi_2,v)=0$ then $1$ is output;
if $\pr_1(\psi_1,v)=0$ and $\pr_2(\psi_2,v)>0$ then $2$ is output;
if $\pr_1(\psi_1,v)=0$ and $\pr_2(\psi_2,v)=0$ then $3$ is output.
In those cases the monitor terminates immediately.

After $N$ phases, if $x\leq 0$ then the monitor~$M_2'$ outputs~$1$, otherwise it outputs~$2$.
An output of~$i$ indicates that the sequence is believed to be generated by~$H_i$.
Note that $M_2'$---in contrast to~$M_2$---needs access to the function~$\A$.
By constructing a supermartingale and applying Azuma's inequality we obtain:
\newcommand{\stmtthmalgone}{
Consider the monitor~$M_2'$ running $N$ phases.
Let $\Lang(M_2') \subseteq \Sigma^\omega$ be the set of observation sequences for which $M_2'$ outputs~$1$.
Then, 
 \begin{align*}
  \P_1(\io{\Lang(M_2')}_1) \ &\ge \ 1-\exp\left(-\frac{c^2}{18}\cdot N\right) \quad \text{and} \\ 
  \P_2(\io{\Lang(M_2')}_2) \ &\le \ \exp\left(-\frac{c^2}{18}\cdot N\right)\,. 
 \end{align*}
}%
\begin{theorem} \label{thm-alg1}
\stmtthmalgone
\end{theorem}
Hence the error probability decays exponentially with~$N$.
To prove Theorem~\ref{thm-alg1-likeli} we show (\iftechrep{in the appendix}{in~\cite{16KS-TR}}) that the same error bound, up to a factor of~$2$, holds for the likelihood-based monitor~$M_2$.
The authors are not aware of a proof of Theorem~\ref{thm-alg1-likeli} that avoids reasoning about a monitor like~$M_2'$.
The proof shows that the difference $\P_1(\io{\Lang(M_2)}_1) -
\P_2(\io{\Lang(M_2)}_2)$ cannot be increased by any other monitor
that is based solely on the first $N \cdot m$ observations: $M_2$ is
optimal in that respect.

To guarantee an error probability bound of at most~$\varepsilon$ of the likelihood-based monitor $M_2$, we set  $N =
\lceil\frac{18}{c^2}\cdot \log
\left(\frac{2}{\varepsilon}\right)\rceil$.

%

%
\begin{example} \label{ex-1}
Figure~\ref{fig-ex-1} shows two HMCs $H_1, H_2$ with a parameter $\delta \in (0, \frac14]$.
\begin{figure*}[ht]
\begin{center}
\begin{tikzpicture}[scale=2.5,LMC style]
\node (H1) at (-1.8,0.1) {\large $H_1 : $};
\node[state] (s0) at (-1,0) {$s_0$};
\node (a0) at (-1,-0.2) {$a$};
\node[state] (s1) at ( 0,0) {$s_1$};
\node (b0) at (-0,-0.2) {$b$};
\node (H2) at (1.2,0.1) {\large $H_2 : $};
\node[state] (t0) at (+2,0) {$t_0$};
\node (a1) at (+2,-0.2) {$a$};
\node[state] (t1) at (+3,0) {$t_1$};
\node (b1) at (+3,-0.2) {$b$};
\path[->] (-1.5,0) edge (s0);
\path[->] ( 1.5,0) edge (t0);
\path[->] (s0) edge [loop,out=110,in=70,looseness=13] node[above] {$\frac12 + \delta$} (s0);
\path[->] (s0) edge [bend left=20] node[above] {$\frac12 - \delta$} (s1);
\path[->] (s1) edge [bend left=20] node[below] {$\frac12 + \delta$} (s0);
\path[->] (s1) edge [loop,out=110,in=70,looseness=13] node[above] {$\frac12 - \delta$} (s1);
\path[->] (t0) edge [loop,out=110,in=70,looseness=13] node[above] {$\frac12 ´- \delta$} (t0);
\path[->] (t0) edge [bend left=20] node[above] {$\frac12 + \delta$} (t1);
\path[->] (t1) edge [bend left=20] node[below] {$\frac12 - \delta$} (t0);
\path[->] (t1) edge [loop,out=110,in=70,looseness=13] node[above] {$\frac12 + \delta$} (t1);
\end{tikzpicture}
\end{center}
\caption{Two distinguishable HMCs with a parameter $\delta \in (0,\frac14]$}
\label{fig-ex-1}
\end{figure*}
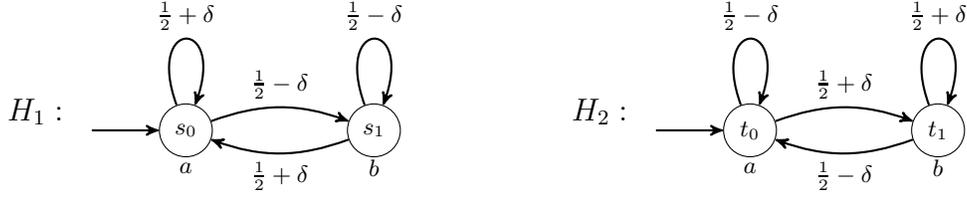
In every step except the first one, $H_1$ outputs~$a$ with probability $\frac12 + \delta$, and $b$ with probability $\frac12 - \delta$.
For $H_2$ the probabilities are reversed.
The HMCs are distinguishable.
The intuitive reason is that $H_1$ tends to output more $a$-symbols than $b$-symbols, whereas $H_2$ tends to output more $b$-symbols than $a$-symbols, and this difference is exhibited in the long run.
Intuitively speaking, the smaller $\delta$ is, the ``less distinguishable'' are $H_1$ and~$H_2$.
We will show later that there is a profile with $c = \delta$.
By Theorem~\ref{thm-alg1-likeli}, the probability that $M_2$ mistakes $H_2$ for~$H_1$ decays exponentially.
More specifically, for an error bound of~$\varepsilon$ it suffices to make $B \delta^{-2} \log\frac{1}{\varepsilon}$ observations, for a constant~$B>0$.
It can be shown that there is a constant $d>0$ such that $M_2$ needs, for small~$\varepsilon$, at least $d \delta^{-2} \log\frac{1}{\varepsilon}$ observations to push the error probability below~$\varepsilon$.
Hence, for the HMCs from Figure~\ref{fig-ex-1}, the bound of Theorem~\ref{thm-alg1-likeli} is asymptotically tight.
As mentioned after the proof of Theorem~\ref{thm-alg1-likeli}, the likelihood-based monitor is essentially optimal among monitors that observe a fixed-length prefix.
So the bound from Theorem~\ref{thm-alg1} is also asymptotically tight.
\qed
\end{example}

\subsection{Monitors with One-Sided Error} \label{sub-monitor-one-sided}
Now we present~$M_1$, a likelihood-based monitor with one-sided error.
Monitor~$M_1$ uses a threshold parameter $\low \in (0,1]$.
For each $N>0$, after reading a prefix $v$, of length~$N
\cdot m$, of observations, it computes the likelihood ratio $\lr(v)$. If
$\lr(v)\leq \low$, it terminates outputting~$1$, otherwise it continues. 

For any infinite sequence $u$ and integer $i>0$, let $u[i]$ denote the prefix of $u$ of length $i$.
We fix an integer $N>0$. 
Let $U_N$ be the set of all $u\in \Sigma^{\omega}$ such that $\lr(u[N\cdot m])\leq \exp\bigl(-\frac{c^2}{36}\cdot N\bigr)$. 
Recall from Theorem~\ref{thm-alg1-likeli} the set
$\Lang(M_2)$ of observation sequences for which $M_2$ outputs $1$. 
It should be easy to see that $U_N\subseteq \Lang(M_2)$.
We need the following technical lemma.
\newcommand{\stmtlemtechnical}{
$\displaystyle\P_1(\io{U_N}_1)\ \geq\ 1-4
\exp\Big(-\frac{c^2}{36}\cdot N\Big)$
}
\begin{lemma}\label{lem-technical}
\stmtlemtechnical
\end{lemma}

This allows us to prove the following theorem:
\newcommand{\stmtthmalgtwolikeli}{
Consider the monitor~$M_1$ with threshold parameter $\low \in (0,1]$.
Let $\Lang(M_1) \subseteq \Sigma^\omega$ be the set of observation sequences for which $M_1$ terminates (and hence outputs~$1$).
Then,
 \begin{align*}
  \P_1(\io{\Lang(M_1)}_1) \  &=  \ 1 \qquad \text{and} \\
  \P_2(\io{\Lang(M_1)}_2) \  &\le \ \low\,.
 \end{align*}
}
\begin{theorem} \label{thm-alg2}
\stmtthmalgtwolikeli
\end{theorem}

Now we analyze the response time of~$M_1$ taken on observation sequences generated by~$H_1$.
Formally, we define a random variable $T: \{s_{1,0}\} S_1^\omega \to
\N$ such that $T$ is the number of observations made by monitor~$M_1$
before outputting~$1$.
The following proposition bounds the expected value of~$T$ in~$H_1$.
\newcommand{\stmtpropalgtwolikelibound}{
$\E_1(T) \le \frac{36 m}{c^2} \cdot \log \frac{1}{\low} + \frac{147
  m}{c^2}\cdot \low +m$,
where $\E_1(T)$ is the expected value of~$T$ under the probability measure~$\P_1$.
}
\begin{proposition} \label{prop-alg2likeli-bound} 
\stmtpropalgtwolikelibound
\end{proposition}
The proof of this proposition employs ideas similar to those in~\cite{Sistla14} for proving an upper bound
on the expected monitoring time for exponentially converging monitorable systems.
Observe that as $\low$ decreases, the first term of the bound dominates.

\subsection{Monitors for Distinguishing Among Multiple HMCs} \label{sub-multiple}
Now we address the problem of distinguishing among multiple mutually distinguishable HMCs. 
We present a monitor based on likelihoods. 
For $i=1,\ldots,k$, let $H_i=(G_i,O_i,s_{i,0})$ be HMCs with the same observation alphabet~$\Sigma$ where $G_i=(S_i,R_i,\phi_i)$.
Let $\P_i$ and $\io{\cdot}_i$ be the associated probability measures and inverse observation functions corresponding to the HMC~$H_i$.
We assume that they are mutually distinguishable, i.e.,
for $1\leq i<j\leq k$, HMCs $H_i$~and~$H_j$ are distinguishable.
So by Proposition~\ref{prop-profile} there are profiles $(\A_{i,j}, c_{i,j})$. 
Define $c := \min \{ c_{i,j} \mid 1 \le i < j \le k\}$.

Let $m := 2\cdot \max\{\card{S_i}\mid 1\leq i\leq k\}$ and $N>0$ be an integer parameter.
The following monitor~$M$ distinguishes among the $k$~HMCs:
it takes an observation sequence $u\in \Sigma^{N\cdot m}$ as input and outputs the smallest integer 
$i \in \{1, \ldots, k\}$ such that
$\pr_i(\delta_{s_{i,0}},u) \ge \pr_j(\delta_{s_{j,0}},u)$ for all $j \in \{1, \ldots, k\}$.
Essentially, $M$ outputs the index of the HMC whose likelihood value is the highest after $N \cdot m$ observations.
By applying the union bound to Theorem~\ref{thm-alg1-likeli} we get:
\newcommand{\stmtthmalgmultiple}{
Consider the monitor~$M$.
Let $i \in \{1, \ldots, k\}$ and
let $\Lang_i\subseteq \Sigma^{N\cdot m}$ be the set of observation sequences for which $M$ outputs~$i$.
Then we have for all $j \in \{1, \ldots, k\} - \{i\}$:
 \begin{align*}
  \P_i(\io{\Lang_i\Sigma^\omega}_i) \ &\ge \ 1-2k\cdot \exp\left(-\frac{c^2}{18}\cdot N\right) \quad \text{and} \\
  \P_j(\io{\Lang_i\Sigma^\omega}_j) \ &\le \ 2 \exp\left(-\frac{c^2}{18}\cdot N\right)
 \end{align*}
}
\begin{theorem}\label{thm-algmultiple}
\stmtthmalgmultiple
\end{theorem}

%% file: profiles.tex
\section{Computing Profiles} \label{sec-profiles}
In the monitors of Section~\ref{sec-monitors} the constant~$c>0$ determines the
number $N$ of phases needed to ensure a bound on the error
probability. Recall that $c$ is the constant in a profile
$(\A,c)$. Any such constant $c$ will do, but the larger it is the
better, since the number of phases used will be smaller.
Note that even the existence of a positive~$c$ (as claimed by Proposition~\ref{prop-profile}) is not obvious.
In this section, we prove Theorem~\ref{thm-profile}---which strengthens Proposition~\ref{prop-profile}---by presenting a polynomial-time algorithm to compute a positive~$c$ and also the representation of a profile function $\A$ in polynomial time.


Let a \emph{test set} $\Test \subseteq \Sigma^*$ be a set of at most~$m$ words, with $|v| < m$ for all $v \in \Test$.
This defines a function $\A_\Test : \Distr{S_1} \times \Distr{S_2} \to 2^{\Sigma^m}$ in the following way.
Fix $\psi_1 \in \Distr{S_1}$ and $\psi_2 \in \Distr{S_2}$.
Let $v \in \Test$ be such that
\begin{equation} \label{eq-arg-max}
v := \arg\max_{w \in \Test} \card{\pr_1(\psi_1,w)-\pr_2(\psi_2,w)}
\end{equation}
and write
\[
 \vv{v} := \{v w \mid w \in \Sigma^*, \ |v w| = m\}
\]
for the set of strings of length~$m$ with $v$ as a prefix.
Then define:
\begin{equation*}
 \A_\Test(\psi_1,\psi_2) := 
 \begin{cases}
   \vv{v} & \text{if $\pr_1(\psi_1,v)>\pr_2(\psi_2,v)$} \\
  \Sigma^m - \vv{v} & \text{otherwise}
 \end{cases}
\end{equation*}
Depending on the case above, $\pr_i(\psi_i,\A_\Test(\psi_1,\psi_2))$ is either $\pr_i(\psi_i,v)$ or $1-\pr_i(\psi_i,v)$.
Hence:
\begin{multline} \label{eq-pos-diff}
\pr_1(\psi_1,\A_\Test(\psi_1,\psi_2)) - \pr_2(\psi_2,\A_\Test(\psi_1,\psi_2)) \\
= \card{\pr_1(\psi_1,v) - \pr_2(\psi_2,v)}
\end{multline}

Given a test set~$\Test$ and distributions $\psi_1, \psi_2$, a monitor can compute the word~$v$ from~\eqref{eq-arg-max} using~\eqref{eq-sub-to-pr}, and hence the probabilities $\pr_i(\psi_i,\A_\Test(\psi_1,\psi_2))$.
Moreover, a monitor can check whether a given word $w \in \Sigma^m$ is in~$\A_\Test(\psi_1,\psi_2)$ by checking whether $v$ is a prefix of~$w$.

\newcommand{\stmtthmprofile}{
Let HMCs $H_1, H_2$ be distinguishable.
One can compute, in polynomial time, a test set $\Test \subseteq \Sigma^*$ and a number $c > 0$ such that
$(\A_\Test,c)$ is a profile.
}
\begin{theorem}\label{thm-profile}
\stmtthmprofile
\end{theorem}
The proof builds on~\cite{14CK-LICS} but requires further insights.
For the proof we need the concept of \emph{equivalence}:
For $i=1,2$ let $\psi_i\in \Distr{S_i}$.
We say that $\psi_1$ is \emph{equivalent} to $\psi_2$, written as $\psi_1\equiv \psi_2$, 
if $\pr_1(\psi_1,u) = \pr_2(\psi_2,u)$ holds for all $u\in \Sigma^*$.
We have the following proposition:
\newcommand{\stmtproptestset}{
One can compute, in polynomial time, a test set $\Test \subseteq \Sigma^*$ such that for all
$\psi_1 \in \Distr{S_1}$
and all $\psi_2 \in \Distr{S_2}$
we have:
\[
 \psi_1\equiv \psi_2 \quad \Longleftrightarrow\quad \forall u\in
 \Test:\ \pr_1(\psi_1,u)=\pr_2(\psi_2,u)
\]
}
\begin{proposition} \label{prop-test-set}
\stmtproptestset
\end{proposition}
The algorithm for Proposition~\ref{prop-test-set} uses linear-algebra based techniques that have been developed for deciding equivalence of HMCs, see e.g.~\cite{Schutzenberger,Tzeng,Ito-Equivalence,Doyen-Equivalence}.

We fix $\Test$ for the remainder of the section.
We define a distance measure $\dist(\psi_1,\psi_2)$ between
$\psi_1,\psi_2$ given by
\[
 \dist(\psi_1,\psi_2) := \max_{w \in \Test} \card{\pr_1(\psi_1,w)-\pr_2(\psi_2,w)}\,.
\]
By Proposition~\ref{prop-test-set} we have:
\[
 \psi_1\equiv \psi_2 \quad \Longleftrightarrow \quad \dist(\psi_1,\psi_2)=0 
\]
For the following proposition, linear programming is used to compute a lower bound on $\dist(\psi_1,\psi_2)$ for reachable pairs $(\psi_1, \psi_2)$ in distinguishable HMCs:
\newcommand{\stmtproplbound}{
Let $H_1,H_2$ be distinguishable HMCs.
One can compute, in polynomial time, a rational number~$c > 0$
such that for all reachable pairs $(\psi_1, \psi_2)$ of distributions we have $\dist(\psi_1, \psi_2) \ge c$.
}
\begin{proposition} \label{prop-lbound}
\stmtproplbound
\end{proposition}
In general there may exist \emph{unreachable} pairs $(\psi_1, \psi_2)$ of distributions with $\dist(\psi_1, \psi_2) = 0$, even for distinguishable HMCs.
Proposition~\ref{prop-lbound} establishes in particular the nontrivial fact that for distinguishable HMCs there \emph{exists} a positive lower bound on~$\dist(\psi_1, \psi_2)$ for all reachable pairs $(\psi_1, \psi_2)$.

\begin{example} \label{ex-2}
Consider again the HMCs from Figure~\ref{fig-ex-1}.
We compute the set~$\Test$ according to the algorithm from Proposition~\ref{prop-test-set}.
This yields $\Test = \{\varepsilon, a, a a, b a\}$, where $\varepsilon$ denotes the empty word.

In this example, the last symbol of any observation sequence reveals the state.
Hence there are only two reachable pairs of distributions:
one is $(\pi_1, \pi_2)$ with $\pi_1(s_0) = \pi_2(t_0) = 1$,
and the other one is $(\pi_1', \pi_2')$ with $\pi_1'(s_1) = \pi_2'(t_1) = 1$.
Using the definition of~$\dist$ we compute:
\begin{align*}
\dist(\pi_1, \pi_2)   & = \pr_1(\pi_1, a a) - \pr_2(\pi_2, a a) \\
                      & = \frac12 + \delta - \left(\frac12 - \delta\right) = 2 \delta \\
\dist(\pi_1', \pi_2') & = \pr_1(\pi_1', b a) - \pr_2(\pi_2', b a) \\
                      & = \frac12 + \delta - \left(\frac12 - \delta\right) = 2 \delta
\end{align*}
Hence we have $\dist(\psi_1, \psi_2) = 2 \delta > 0$ for all \emph{reachable} pairs $(\psi_1, \psi_2)$ of distributions.

In order to illustrate some aspects of Proposition~\ref{prop-lbound}, we use linear programming to compute a lower bound on $\dist(\psi_1, \psi_2)$ for \emph{all} (reachable or unreachable) pairs $(\psi_1, \psi_2)$ of distributions.
Concretely, we solve the following linear program, where $\delta$ is the constant parameter from the HMCs $H_1, H_2$, and the variables are $x$ and variables encoding distributions $\psi_1, \psi_2$:
\begin{align*}
 \text{minimize }    & x \ge 0  \\ 
 \text{subject to: } & \psi_1 \in \Distr{S_1}, \ \psi_2 \in \Distr{S_2}, \\
                     & -x \le \pr_1(\psi_1, u) {-} \pr_2(\psi_2, u) \le x  \text{ for all $u \in \Test$.} \nonumber
\end{align*}
An optimal solution is $x = \delta$ and 
$\psi_1(s_0) = \frac34 - \frac{\delta}{2}$ and
$\psi_1(s_1) = \frac14 + \frac{\delta}{2}$ and
$\psi_2(t_0) = \frac34 + \frac{\delta}{2}$ and
$\psi_2(t_1) = \frac14 - \frac{\delta}{2}$.
Hence $x = \delta > 0$ is a lower bound on~$\dist(\psi_1, \psi_2)$ for \emph{all} pairs of distributions,
and hence, a fortiori, also for all reachable pairs.
As mentioned after Proposition~\ref{prop-lbound}, the reachability aspect is in general (unlike in this example) essential for obtaining a positive lower bound.
Indeed, the proof of Proposition~\ref{prop-lbound} takes advantage of further results from~\cite{14CK-LICS}.

If we compute a lower bound according to the proof Proposition~\ref{prop-lbound}, i.e., taking reachability into account, we obtain $c = 4 \delta / (3 + 2 \delta)$, which lies strictly between the previously computed lower bounds $\delta$ and~$2 \delta$.
\qed
\end{example}

With Proposition~\ref{prop-lbound} at hand, we are ready to prove Theorem~\ref{thm-profile}:
\begin{proof}[of Theorem~\ref{thm-profile}]
Compute $\Test$ according to Proposition~\ref{prop-test-set} and $c>0$ according to Proposition~\ref{prop-lbound}.
We show that $(\A_\Test,c)$ is a profile.
Let $(\psi_1, \psi_2)$ be a reachable pair of distributions.
Let $v \in \Sigma^*$ be as in~\eqref{eq-arg-max}.
We have:
\begin{align*}
& \ \rlap{$\pr_1(\psi_1,\A_\Test(\psi_1,\psi_2)) - \pr_2(\psi_2,\A_\Test(\psi_1,\psi_2))$} \\
& = \card{\pr_1(\psi_1,v) - \pr_2(\psi_2,v)} && \text{by \eqref{eq-pos-diff}} \\
& = \max_{w \in \Test} \card{\pr_1(\psi_1,w)-\pr_2(\psi_2,w)} && \text{by \eqref{eq-arg-max}} \\
& = \dist(\psi_1,\psi_2)                     && \text{def.\ of $\dist$} \\
& \ge c                                      && \text{Proposition~\ref{prop-lbound}}
\end{align*}
This completes the proof.
\end{proof}

We have seen that for a given error bound, the number of observations our monitors need to make depends quadratically on~$\frac{1}{c}$.
So it may be beneficial to compute a larger value of~$c$, even if such a computation is expensive.
To this end, for a distribution $\pi \in \Distr{S}$, write $\supp(\pi) := \{s \in S \mid \pi(s)>0\}$.
For HMCs $H_1, H_2$, if a pair $(\psi_1, \psi_2)$ of distributions is reachable, we say that the pair
 $(\supp(\psi_1), \supp(\psi_2))$ is \emph{reachable}.
We have the following proposition:
\newcommand{\stmtpropexpensive}{
Let $H_1, H_2$ be two distinguishable HMCs.
One can compute, in exponential time:
\begin{align*}
 c \ := \ & \min_{\text{reachable } (S_1', S_2') \in 2^{S_1} \times 2^{S_2}}  \
          \min_{\psi_1 \in \Distr{S_1'}} \
          \min_{\psi_2 \in \Distr{S_2'}} \\
          & \quad \max_{U \subseteq \Sigma^m} \
          \big( \pr_1(\psi_1, U) - \pr_2(\psi_2, U) \big)
\end{align*}
}
\begin{proposition} \label{prop-expensive-c}
\stmtpropexpensive
\end{proposition}
(Note that $U$ ranges over a set of double-exponential size.)
This value of~$c$ is lower-bounded by the value of~$c > 0$ from Theorem~\ref{thm-profile},
and it is part of a profile with 
\[
 \A(\psi_1, \psi_2) \ = \ \arg \max_{U \subseteq \Sigma^m} \big( \pr_1(\psi_1, U) - \pr_2(\psi_2, U) \big)\,.
\]

%% file: verification.tex
\section{Application: Runtime Verification}
\label{sec-verification}
In this section we discuss an application of monitors for runtime verification of stochastic systems.
Traditional verification aims at proving correctness of systems at the time of their design.
This quickly becomes infeasible, in particular for complex systems with several components and stochastic behavior, see e.g.~\cite{Sistla11}.
\emph{Runtime verification} is an alternative where a monitor observes a system while it is running, and raises an alarm once a faulty behavior is detected.
The alarm may trigger, e.g., a fail-safe way of shutting the system down.
HMCs were suggested in~\cite{Sistla11,Sistla14} as models of partially observable stochastic systems.
In this section, the monitor does not try to distinguish two HMCs, rather it tries to distinguish correct and faulty behavior of a single~HMC.

\myparagraph{Definitions.}
For a probability measure~$\P$ and measurable sets $C,D$ such that $\P(C)>0$, we let $\P(D \mid C )$ denote the value $\frac{\P(C\cap D)}{\P(C)}$, which is the conditional probability of~$D$ given~$C$.
A \emph{classifying HMC (cHMC)} is a quadruple $H = (G, O, s_0, \Class)$, 
where $(G, O, s_0)$ is an HMC
and $\Class$ is a condition classifying each bottom strongly connected component (BSCC) of~$H$
  as \emph{bad} or \emph{good}.
For a cHMC and a state $s \in S$ we define:
\begin{align*}
 \Bad_s  &:= \{s s_1 s_2 \cdots \in \{s\} S^\omega \mid \exists i: s_i \text{ is in a bad BSCC}\} \\
 \Good_s &:= \{s s_1 s_2 \cdots \in \{s\} S^\omega \mid \exists i: s_i \text{ is in a good BSCC}\}
\end{align*}
Define $\Bad := \Bad_{s_0}$ and $\Good := \Good_{s_0}$.
The events $\Bad$~and~$\Good$ are disjoint and measurable.
By fundamental properties of Markov chains we have
\[
 \P(\Bad \cup \Good) = \P(\Bad) + \P(\Good) = 1\,.
\]
To avoid trivialities we assume that $\P(\Bad), \P(\Good) > 0$ (this can be checked in polynomial time by graph reachability).
We say that a cHMC~$H$ is \emph{monitorable}
 if for every $\varepsilon > 0$ there exists a monitor~$M$ such that
\begin{align*}
  \P(\io{\Lang(M)} \mid \Bad) \ &\ge \ 1-\varepsilon \quad \text{ and  } \\
  \P(\io{\Lang(M)} \mid \Good) \ &\le \ \varepsilon\,.
\end{align*}
In~\cite{Sistla11} the authors define and study monitorability of pairs $(H_0, \A)$
 where $H_0$ is an HMC and $\A$ is a deterministic Streett automaton.
One can compute, in polynomial time, the product of $H_0$ and~$\A$.
That product is a cHMC~$H$ as defined above.
Then $(H_0, \A)$ is monitorable (in the sense of~\cite{Sistla11}) if and only if
 $H$ is monitorable (in the sense defined above).

A construction similar to one that was given in~\cite[Section 3]{BKKM14conditional} allows us, for a given cHMC~$H$, to construct two HMCs $H_1, H_2$ that exhibit the bad and the good behavior of~$H$ according to their conditional probabilities:

\newcommand{\stmtpropconditioned}{
Let $H$ be a cHMC with $\P(\Bad), \P(\Good) > 0$.
Then one can compute, in polynomial time,
 HMCs $H_1, H_2$ such that
 for all measurable events $E \subseteq S^\omega$ we have
\[
 \P_1(E) = \P(E \mid \Bad) \quad \text{and} \quad
 \P_2(E) = \P(E \mid \Good)\,.
\]
}
\begin{proposition} \label{prop-conditioned}
\stmtpropconditioned
\end{proposition}

It follows from Proposition~\ref{prop-conditioned}
 that distinguishing and monitoring are equivalent:
Given HMCs $H_1, H_2$, we can combine them into a single cHMC~$H$ by introducing a new initial state~$s_0$, which branches to the initial states of $H_1, H_2$ with probability $1/2$ each.
We classify the BSCCs of~$H_1$ and of $H_2$ as bad and good, respectively.
Then for any $E \subseteq \Sigma^\omega$ we have
\begin{align*}
 \P_1(E) & = \P(\{O(s_0)\} E \mid \Bad) \quad \text{ and } \\
 \P_2(E) & = \P(\{O(s_0)\} E \mid \Good)\,,
\end{align*}
 so any monitor for~$H$ can be translated in a straightforward way into a monitor that distinguishes $H_1$ and~$H_2$.
Conversely, given a cHMC~$H$, we can compute $H_1,H_2$ according to Proposition~\ref{prop-conditioned}.
Then any monitor that distinguishes $H_1$ and~$H_2$ also monitors~$H$.

By combining this observation with Theorem~\ref{thm-qual-P} we obtain:
\begin{corollary} \label{cor-qual-P-cHMC}
One can decide in polynomial time whether a given cHMC~$H$ is monitorable.
\end{corollary}
Another kind of monitorability, called \emph{strong monitorability}~\cite{Sistla11}, was shown PSPACE-complete in~\cite{Sistla11}.
Strong monitorability implies monitorability.

Using Proposition~\ref{prop-conditioned} again, the monitors from Section~\ref{sec-monitors} apply to monitoring cHMCs.
For instance, the monitor with one-sided error can guarantee that (a) given that the behavior is faulty then an alarm is raised with probability~$1$ and within short expected time, and (b) given that the behavior is correct then probably no alarm is raised.

%% file: conclusions.tex
\section{Conclusions}
\label{sec-conclusions}
In this paper we have considered the distinguishability problem for HMCs.
We have shown that it is decidable in polynomial time.

We have presented two likelihood based monitors~$M_1, M_2$ for distinguishing between HMCs $H_1, H_2$ based on the sequences of observations generated by them.
The monitor~$M_2$ makes a decision after running for a fixed number of observations and exhibits two-sided error.
It processes $O(\log \frac{1}{\varepsilon})$ observations to ensure an error probability of at most~$\varepsilon$.
The monitor~$M_1$ has only one-sided error.
The expected number of observations it processes to identify a sequence generated by~$H_1$ is $O(\log \frac{1}{\varepsilon})$ to guarantee an error probability of at most~$\varepsilon$ on sequences generated by~$H_2$.
We have also provided a monitor for distinguishing multiple HMCs.
All error analyses rely on martingale techniques, in particular, Azuma's inequality.

Polynomial time bounded algorithms are provided, which for the monitor~
$M_2$, compute the number of observations that guarantees a given
upper bound on the error, and for the~$M_1$ compute the expected
number of observations of $H_1$ before which an alarm is raised, for a
given error bound on the probability of raising an alarm on inputs
generated by $H_2$.  
These algorithms  employ  linear programming based techniques for
computing profiles.

We have discussed an application to runtime verification of stochastic systems.
The monitorability problem for cHMCs is polynomial-time equivalent to distinguishability, and hence decidable in polynomial time.
We have shown that the monitors developed in this paper can be adapted so that they monitor cHMCs.

One direction for future work is to improve the efficiency of computing a good lower bound on~$c$.
We have seen that this bound strongly influences the number of observations the monitor needs to make,
so the bound may determine the applicability of a monitor in practice.
Another direction is to develop a notion of a monitor for HMCs that are not equivalent but not distinguishable.
Such monitors might still attempt to distinguish between the HMCs for as many runs as possible.

%% file: app-monitors.tex
\section{Proofs of Section~4} \label{app-monitors}

We prove Theorem~\ref{thm-alg1} from the main text.

\begin{qtheorem}{\ref{thm-alg1}}
\stmtthmalgone
\end{qtheorem}
\begin{proof}
Let $\psi_{1,k},\psi_{2,k}$ and $X_k$ denote the values of
$\psi_1,\psi_2$ and $x$ directly after the $k$-th phase, for $k\geq 0$.
Initially, $\psi_{1,0}=\delta_{s_{1,0}}$, $\psi_{2,0}=\delta_{s_{2,0}}$ and $X_0=0$.
For $k \ge 1$, let $u_k \in \Sigma^{k \cdot m}$ be the sequence of all observations received until and including phase~$k$.
By induction on $k$, it is easy to see that $\psi_{i,k}\:=\cd_i(\delta_{s_{i,0}},u_{k})$ for $k \ge 0$, $i=1,2$.

Note that $X_k$ depends only on~$u_k$.
In the following we view $X_0, X_1, \ldots$ as a sequence of random variables.
(Formally, for \mbox{$k \ge 0$} the random variable $X_k$ is a function of type $X_k: \{s_{1,0}\} S_1^\omega \to \Q$.)
We also define a sequence $Y_0, Y_1, \ldots$ of random variables with $Y_k = X_k + k \cdot \frac{c}{2}$.
Note that $X_0 = Y_0 =0$.

We show that the sequence of random variables $Y_0, Y_1, \ldots$ forms a supermartingale in~$H_1$.
Let $k \ge 0$.
Fix $u_k \in \Sigma^{k \cdot m}$.
Recall that this determines~$X_k$.
For the conditional expected value of $X_{k+1}$ given prefix~$u_k$ we have:
\begin{equation} \label{eq-alg1-mart-X-1}
 \E (X_{k+1} \mid \io{u_k \Sigma^\omega}_1) \ = \ X_k + d\;,
\end{equation} 
where $d$ denotes the expected change of~$x$ after phase~$k+1$.
Recall that $\psi_{i,k}= \cd_i(\delta_{s_{i,0}},u_{k})$ for $i=1,2$.
Let $p_i =  \pr_i(\psi_{i,k},\A(\psi_{1,k},\psi_{2,k}))$, for $i=1,2$.
Assume $p_1+p_2\leq 1$.
According to our rule for updating~$x$ we then have:
\begin{align*}
d & \ = \ p_1\cdot (-1) + (1-p_1)\cdot \frac{p_1+p_2}{2-p_1-p_2} \ = \ \frac{p_2-p_1}{2 - p_1-p_2}
\intertext{This is negative. Moreover, by the definition of a profile
    we have $p_1 - p_2 \ge c>0$. Further more, $1\leq 2-p_1-p_2<2$.
Hence:}
d & \ \le \ \frac{p_2-p_1}{2} \ \le \ - \frac{c}{2}
\end{align*}
Combining this with~\eqref{eq-alg1-mart-X-1} and the definition of~$Y_k$ we obtain:
\begin{equation} \label{eq-alg1-mart-Y-1}
\E (Y_{k+1} \mid \io{u_k \Sigma^\omega}_1) \ = \ Y_k + \frac{c}{2} + d \ \le \ Y_k
\end{equation}
Now assume $p_1 + p_2 > 1$.
Then we have:
\begin{align*}
d & \ = \ -p_1\cdot \frac{2-p_1-p_2}{p_1+p_2} + (1-p_1) \cdot 1 \ = \ \frac{p_2 - p_1}{p_1+p_2} \\
  & \ \le \ \frac{p_2 - p_1}{2} \ \le \ - \frac{c}{2} \;,
\end{align*}
so \eqref{eq-alg1-mart-Y-1} again follows.
Hence we have shown that $Y_0, Y_1, \ldots$ is a supermartingale in~$H_1$.

By definition of the update rule we have $\card{X_{k+1} - X_k} \le 1$ and hence $\card{Y_{k+1} - Y_k} \le 1 + \frac{c}{2} \le \frac{3}{2}$.
Applying Azuma's inequality (see, e.g.,~\cite{book:Williams}) to the supermartingale $Y_0,Y_1, \ldots$ we obtain:
\begin{align*}
 \P_1\left\{X_N>0\right\} \ & = \ \P_1\left\{Y_N>\frac{c}{2}\cdot N\right\} \ \leq \ \exp\left(-\frac{\left(\frac{c}{2}\cdot N\right)^2}{2 N \cdot \left(\frac{3}{2}\right)^2}\right) \\ \ & = \ \exp\left(-\frac{c^2}{18}\cdot N\right)
\end{align*}
Hence,
\[
\P_1(\io{\Lang(M)}_1)
 \ = \ \P_1\{X_N\leq 0\}
 \ \geq \ 1-\exp\left(-\frac{c^2}{18}\cdot N\right).
\]
From this, it follows that  $\P_1(\io{\Lang(M)}_1) \ge 1- \exp(-\frac{c^2}{18}\cdot 
N).$

The proof of the second inequality in the statement is similar with the following modifications.
The random variables $X'_k$ are defined like~$X_k$, but on sequences of states in~$H_2$ rather than~$H_1$.
Define $Y'_k = X'_k - k \cdot \frac{c}{2}$.
The sequence $Y'_0, Y'_1, \ldots$ is now a submartingale.
Applying Azuma's inequality to this submartingale now leads to the second inequality claimed in the statement.
\end{proof}

The following lemma is used for the proof of Theorem~\ref{thm-alg1-likeli}.

\begin{lemma}\label{lem-maximizing-event-countable}
Let $S$ be a countable set.
Let $\psi_1, \psi_2$ be probability distributions over~$S$.
For $i \in \{1,2\}$ and any event $V \subseteq S$ define $\psi_i(V) := \sum_{v \in V} \psi_i(v)$.
Define
\[
 W := \{s \in S \mid \psi_1(s) \ge \psi_2(s)\}\,.
\]
Then
\[
 \max_{V \subseteq S} \big( \psi_1(V) - \psi_2(V) \big) \ = \ \psi_1(W) - \psi_2(W)\,,
\]
i.e., $W$ maximizes the probability difference over all events.
\end{lemma}
\begin{proof}
If $s \in S$ with $s \not\in V$ and $\psi_1(s) \ge \psi_2(s)$, then
\begin{align*}
& \psi_1(V \cup \{s\}) - \psi_2(V \cup \{s\}) \\
& = \psi_1(V) - \psi_2(V) + \psi_1(s) - \psi_2(s) \\
& \ge \psi_1(V) - \psi_2(V)\,.
\end{align*}
Similarly, if $s \in S$ with $s \in V$ and $\psi_1(s) < \psi_2(s)$, then
\begin{align*}
& \psi_1(V \setminus \{s\}) - \psi_2(V \setminus \{s\}) \\
& = \psi_1(V) - \psi_2(V) - \psi_1(s) + \psi_2(s) \\
& > \psi_1(V) - \psi_2(V)\,.
\end{align*}
The statement of the lemma follows.
\end{proof}

Now we prove Theorem~\ref{thm-alg1-likeli} from the main text.

\begin{qtheorem}{\ref{thm-alg1-likeli}}
\stmtcoralgonelikeli
\end{qtheorem}
\begin{proof}
Let $N \ge 0$.
We can write $\Lang(M_2) = W \Sigma^\omega$ where $W \subseteq \Sigma^{N \cdot m}$ denotes the set of observation prefixes of length~$N \cdot m$ on which $M_2$ outputs~$1$.
Then we have:
\begin{align*}
 W \ &= \ \{ u \in \Sigma^{N \cdot m} \mid \pr_1(\delta_{s_{1,0}}, u) \ge \pr_2(\delta_{s_{2,0}}, u) \} \\
     &= \ \{ u \in \Sigma^{N \cdot m} \mid \P_1(\io{\{u\} \Sigma^\omega}_1) \ge \P_2(\io{\{u\} \Sigma^\omega}_2) \}
\end{align*}
(We left the output of the monitor unspecified when the likelihood ratio is equal to~$1$.
As a consequence, the inequalities above might be strict.
This does not affect the rest of the argument.)
Using Lemma~\ref{lem-maximizing-event-countable}
we obtain the following inequality.
\begin{equation} \label{eq-thm-alg1-likeli-1}
\P_1(\io{W \Sigma^\omega}_1) - \P_2(\io{W \Sigma^\omega}_2) \ \ge \ 
\P_1(\io{V \Sigma^\omega}_1) - \P_2(\io{V \Sigma^\omega}_2) 
\end{equation}
for all $V \subseteq \Sigma^{N \cdot m}$.
In particular, this holds for the prefixes of length~$N \cdot m$ of~$\Lang(M_2')$ from Theorem~\ref{thm-alg1}.
Hence we have:
\begin{align*}
& \quad \P_1(\io{\Lang(M_2)}_1) - \P_2(\io{\Lang(M_2)}_2) \\
& = \quad \P_1(\io{W \Sigma^\omega}_1) - \P_2(\io{W \Sigma^\omega}_2) && \text{as $\Lang(M_2) = W \Sigma^\omega$} \\
& \ge \quad \P_1(\io{\Lang(M_2')}_1) - \P_2(\io{\Lang(M_2')}_2) && \text{by~\eqref{eq-thm-alg1-likeli-1}} \\
& \ge \quad 1 - 2 \exp\left(-\frac{c^2}{18}\cdot N\right) && \text{by Theorem~\ref{thm-alg1}}
\end{align*}
This concludes the proof of the Theorem~\ref{thm-alg1-likeli}.
\end{proof}

\myparagraph{Computation for Example~\ref{ex-1}.}
We analyse the likelihood-based monitor~$M_2$ for the HMCs of Figure~\ref{fig-ex-1}.
The monitor~$M_2$ makes $N \cdot m = 4 N$ observations.
It is easy to see that it outputs~$1$ if and only if it reads at least as many $a$-symbols as $b$-symbols, 
i.e., the number of read $a$-symbols is at least $2 N$.
Hence we have:
\begin{align*}
& \quad \P_2(\io{\Lang(M_2)}_2) \\
& = \sum_{i = 2 N}^{4 N} \binom{4 N}{i} \left(\frac12 - \delta\right)^i \left(\frac12 + \delta\right)^{4 N - i} \\
& \ge \binom{4 N}{2 N} \left(\frac12 - \delta\right)^{2 N} \left(\frac12 + \delta\right)^{2 N} \\
&  =  \frac{(4 N)!}{(2 N)! \cdot (2 N)!} \left(\frac14 - \delta^2\right)^{2 N} \\
&  =  \frac{2^{2 N} \cdot (4 N - 1) \cdot (4 N - 3) \cdots 5 \cdot 3}{(2 N) \cdot (2 N - 1) \cdot (2 N - 2) \cdots 2 \cdot 1} \left(\frac14 - \delta^2\right)^{2 N} \\
& \ge \frac{2^{2 N}}{(2 N)} \cdot 2^{2 N - 1} \cdot \left(\frac14 - \delta^2\right)^{2 N} \\
&  = \frac{1}{4 N} \cdot \left(1 - 4 \delta^2\right)^{2 N}
\intertext{For $x \in [0, \frac12]$ we have $\ln (1 - x) \ge - 2 x$. So we can continue as follows:}
& \quad \P_2(\io{\Lang(M_2)}_2) \\
& \ge \frac{1}{4 N} \cdot \exp\left(-16 \delta^2 N\right) \\
& \ge \exp\left(-17 \delta^2 N\right) \qquad \text{for large~$N$}
\end{align*}
It follows that for small~$\varepsilon$, an inequality $\varepsilon \ge \P_2(\io{\Lang(M_2)}_2)$ implies that $N \ge \frac{1}{17} \delta^{-2} \ln \frac{1}{\varepsilon}$.
This completes the calculation for the example.
\qed

\medskip
We prove Lemma~\ref{lem-technical} from the main text.
\begin{qlemma}{\ref{lem-technical}}
\stmtlemtechnical
\end{qlemma}
\begin{proof}
By contradiction. Contrary to the lemma, assume:
\begin{equation}\label{condition-0} 
\P_1(\io{U_N}_1)< 1-4
\exp\Big(-\frac{c^2}{36}\cdot N\Big)
\end{equation}
Let $V_N := \Lang(M_2)-U_N$, and let $W_N$ denote the set of all prefixes, of length $N\cdot m$, of sequences in $V_N$. 
Clearly, for all $v\in W_N$:
\begin{equation*} 
\exp\Big(-\frac{c^2}{36}\cdot N\Big)<\lr(v)<1
\end{equation*}
It follows for all $v\in W_N$:
$$\pr_1(\delta_{s_{1,0}},v)=\frac{\pr_2(\delta_{s_{2,0}},v)}{\lr(v)}<
\exp\Big(\frac{c^2}{36}\cdot N\Big) \cdot \pr_2(\delta_{s_{2,0}},v)$$
\begin{multline}
\text{\hspace{-3.3mm}Hence, \ }
\P_1(\io{V_N}_1)=\sum_{v\in W_N}
\pr_1(\delta_{s_{1,0}},v) 
\\
< \exp\Big(\frac{c^2}{36}\cdot N\Big)\cdot \underbrace{\sum_{v\in W_N} \pr_2(\delta_{s_{2,0}},v)}_{=\P_2(\io{V_N}_2)}\,. \label{condition-2}
\end{multline}
From Theorem~\ref{thm-alg1-likeli} we know that
\[ 
\P_2(\io{V_N}_2)
\ \le\ \P_2(\io{\Lang(M_2)}_2)
\ \le\ 2 \exp\Big(-\frac{c^2}{18}\cdot N\Big)\,.
\]
By combining this with~\eqref{condition-2}, we get:
\begin{equation}
\begin{aligned}
 \P_1(\io{V_N}_1)&<2 \exp\Big(\frac{c^2}{36}\cdot
  N\Big)\cdot \exp\Big(-\frac{c^2}{18}\cdot N\Big) \\
& = 2 \exp\Big(-\frac{c^2}{36}\cdot N\Big)  
\end{aligned} \label{condition-3}
\end{equation}
We have
$\P_1(\io{\Lang(M_2)}_1)= \P_1(\io{V_N}_1) +
\P_1(\io{U_N}_1)$. Using \eqref{condition-0} and~\eqref{condition-3}, we get:
$$\P_1(\io{\Lang(M_2)}_1)< 1-2 \exp\Big(-\frac{c^2}{36}\cdot
  N\Big) < 1-2\exp\Big(-\frac{c^2}{18}\cdot
  N\Big)$$
But this contradicts Theorem~\ref{thm-alg1-likeli}.
\end{proof}

\medskip
We prove Theorem~\ref{thm-alg2} from the main text:
\begin{qtheorem}{\ref{thm-alg2}}
\stmtthmalgtwolikeli
\end{qtheorem}
\begin{proof}
Let $N_0$ be the smallest integer such that
$\exp\left(-\frac{c^2}{36}\cdot N_{0}\right)\leq \low$. Clearly,
for all $N\geq N_0$ we have $\Lang(M_1) \supseteq U_N$ where $U_N$ is the set
defined at the beginning of Section~\ref{sub-monitor-one-sided}. From this observation and Lemma~\ref{lem-technical},
we see that for all $N \ge N_0$:
$$\P_1(\io{\Lang(M_1)}_1)\geq   1-4 \exp\left(-\frac{c^2}{36}\cdot 
  N\right)$$
From this, we get
$$\P_1(\io{\Lang(M_1)}_1)\geq  \lim_{N\to \infty}  1-4 \exp\left(-\frac{c^2}{36}\cdot 
  N\right)\:=1.$$
Let $X\:=\{v\in (\Sigma^m)^*\mid  \pr_1(\delta_{s_{1,0}},v)>0,\: \lr(v)\leq
\low,\: \forall\, i<\card{v}: \lr(v[i])>\low\}$. Intuitively, $X$ is
the set of shortest observation sequences whose length is a multiple
of $m$ and whose likelihood ratio is $\leq \low$. It is easy to see that
$\Lang(M_1)\:=X\Sigma^{\omega}$.
Observe that there do
not exist two distinct sequences $v_1,v_2 \in X$ such that $v_1$ is a prefix of $v_2$.  
\begin{align*}
\P_2(\io{\Lang(M_1)}_2)&=\sum_{v\in X}\pr_2(\delta_{s_{2,0}},v)\\
&\leq \low\cdot \sum_{v\in X}\pr_1(\delta_{s_{1,0}},v)\\
&=\low\cdot \P_1(\io{\Lang(M_1)}_1)\leq \low
\end{align*}
\end{proof}

We prove Proposition~\ref{prop-alg2likeli-bound} from the main text.
\begin{qproposition}{\ref{prop-alg2likeli-bound}}
\stmtpropalgtwolikelibound
\end{qproposition}
\begin{proof}
Since $T$ is a nonnegative integer valued
random variable, from \cite{pp02}, we see that $\E_1(T) = \sum_{n\geq
  0} \P_1\{T>n\}$. Since $M_1$ only decides after each phase, i.e.,
after reading each successive sequence of $m$ observations, we see that 
 $\E_1(T) = \sum_{N\geq
  0} m\cdot \P_1\{T>N\cdot m\}$. Let $N_0$ be the smallest integer
such that  $ \exp\left(-\frac{c^2}{36}\cdot 
  N_0\right)\leq \low$, i.e.,  $N_0=\lceil\frac{36}{c^2}\cdot \log \frac{1}{\low}\rceil$.
We have:
\begin{align}
\E_1(T) &=  \sum_{N=0}^{N_0-1} m\cdot \underbrace{\P_1\{T>N\cdot m\}}_{\le 1}+\sum_{N\geq
 N_0} m\cdot \P_1\{T>N\cdot m\} \nonumber\\
&\leq m\cdot N_0 + m\cdot \sum_{N\geq
 N_0} \P_1\{T>N\cdot m\} 
\label{prop-alg2like-cond1}
\end{align}
For $N \geq 0$, let $X_N\:=\{u\in \Sigma^{\omega} \mid \lr(u[N\cdot
m])>\low\}$.
Observe that, for $N\geq N_0$, $X_N\subseteq \Sigma^{\omega}-U_N$.
Further,
\begin{equation}
\begin{aligned}
\P_1\{T>N\cdot m\}&\leq \P_1(\io{X_N}_1)\leq
                    \P_1(\io{\Sigma^{\omega}-U_N}_1)\\
&\leq 4 \exp\left(-\frac{c^2}{36}\cdot 
  N\right) \text{\quad from Lemma \ref{lem-technical}.}
\end{aligned}
\label{prop-alg2like-cond2}
\end{equation}
From \eqref{prop-alg2like-cond1} and \eqref{prop-alg2like-cond2}, we
get
\begin{align}
\E_1(T) &=  m\cdot N_0 +4m\cdot \sum_{N\geq
 N_0}  \exp\left(-\frac{c^2}{36}\cdot 
  N\right)\nonumber\\
&=m\cdot N_0 +4m\cdot \sum_{N\geq 0}  \exp\left(-\frac{c^2}{36}\cdot 
  (N+N_0)\right)\nonumber\\
&=m\cdot N_0 + 4m\cdot \exp\left(-\frac{c^2}{36}\cdot N_0\right)\cdot \sum_{N\geq
 0}  \exp\left(-\frac{c^2}{36}\cdot 
  N\right)\nonumber\\
&\leq m\cdot N_0 + 4m\cdot \low\cdot
  \frac{1}{1-\exp\left(-\frac{c^2}{36}\right)}\nonumber\\
&\leq \frac{36m}{c^2}\cdot \log \frac{1}{\low} +m+ 4m\cdot \low\cdot
  \frac{1}{1-\exp\left(-\frac{c^2}{36}\right)}\label{prop-alg2like-cond3}\\
 &\text{substituting for $N_0$.}\nonumber
\end{align}
By using a Taylor series expansion of
$\exp\left(-\frac{c^2}{36}\right)$, we get an infinite sum in which the
signs of the terms alternate starting with a positive sign, and in which
the absolute values of the terms decrease monotonically. Hence we can
upper bound its value by the sum of the first three terms, which is
$(1-\frac{c^2}{36}+\frac{c^4}{2\cdot 36^2})$.
From this, we see that $1-\exp\left(-\frac{c^2}{36}\right) \geq
\frac{c^2}{36}-\frac{c^4}{2\cdot 36^2} = \frac{72 c^2 - c^4}{2 \cdot 36^2}$. Using this, after simplification, we see that
\begin{align}  
4m\cdot \low\cdot\frac{1}{1-\exp\left(-\frac{c^2}{36}\right)}&\leq
                                                              \frac{8\cdot
                                                              36^2\cdot
                                                              m\cdot
                                                              \low}{c^2\cdot
                                                              (72-c^2)}\nonumber\\
&\leq\frac{8\cdot 36^2\cdot m\cdot \low}{71\cdot c^2}&\text{since
                                                      $c<1$}\nonumber\\
&\leq \frac{147\cdot m\cdot \low}{c^2}\label{prop-alg2like-cond4}
\end{align}
Using the bound of \eqref{prop-alg2like-cond4} in~\eqref{prop-alg2like-cond3},
we obtain the statement.
\end{proof}

We prove Theorem~\ref{thm-algmultiple} from the main text.
\begin{qtheorem}{\ref{thm-algmultiple}}
\stmtthmalgmultiple
\end{qtheorem}
\begin{proof}
For $j \in \{1, \ldots, k\} - \{i\}$ define:
\[
\Lang_{i,j} := \begin{cases}
\{u\in \Sigma^{N\cdot m}\mid \pr_i(\delta_{s_{i,0}},u) \geq \pr_j(\delta_{s_{j,0}},u)\} & \text{ if $i<j$} \\
\{u\in \Sigma^{N\cdot m}\mid \pr_i(\delta_{s_{i,0}},u) > \pr_j(\delta_{s_{j,0}},u)\}          & \text{ if $i>j$}
\end{cases}
\]
By Theorem~\ref{thm-alg1-likeli} we have:
\begin{equation} \label{eq-thm-algmultiple-M2prime}
 \begin{aligned} 
  \P_i(\io{\Lang_{i,j}\Sigma^\omega}_i) \ &\ge \ 1-2\cdot \exp\left(-\frac{c^2}{18}\cdot N\right) \quad \text{and} \\
  \P_j(\io{\Lang_{i,j}\Sigma^\omega}_j) \ &\le \ 2\exp\left(-\frac{c^2}{18}\cdot N\right)
 \end{aligned}
\end{equation}
We have:
\begin{align*}
& \ 1 - \P_i(\io{\Lang_i\Sigma^\omega}_i) \\
& \  =  \ 1 - \P_i(\io{\cap_{j\neq i}\Lang_{i,j}\Sigma^\omega}_i)
     && \text{$\Lang_i= \cap_{j\neq i}\Lang_{i,j}$} \\
& \  =  \ \P_i(\io{\cup_{j\neq i} (\Sigma^{N\cdot m}-\Lang_{i,j}) \Sigma^\omega}_i) \\
& \ \le \ \sum_{j\neq i}\P_i(\io{(\Sigma^{N\cdot m}-\Lang_{i,j})\Sigma^\omega}_i)
     && \text{union bound} \\
& \  =  \ \sum_{j\neq i} \big( 1 - \P_i(\io{\Lang_{i,j}\Sigma^\omega}_i) \big) \\
& \ \le \ \sum_{j\neq i} \left( 2\cdot \exp\left(-\frac{c^2}{18}\cdot N\right) \right)
     && \text{by~\eqref{eq-thm-algmultiple-M2prime}} \\
& \ \le \ 2 k \cdot \exp\left(-\frac{c^2}{18}\cdot N\right)
\end{align*}
The first inequality follows.
Further we have:
\begin{align*}
\P_j(\io{\Lang_i\Sigma^\omega}_j)
& \ \le \ \P_j(\io{\Lang_{i,j}\Sigma^\omega}_j)
     && \text{$\Lang_i \subseteq \Lang_{i,j}$} \\
& \ \le \ 2 \exp\left(-\frac{c^2}{18}\cdot N\right)
     && \text{by~\eqref{eq-thm-algmultiple-M2prime}}
\end{align*}
This proves the second inequality.
\end{proof}

%% file: app-profiles.tex
\section{Proofs of Section~5} \label{app-profiles}

We prove Proposition~\ref{prop-test-set} from the main text.

\begin{qproposition}{\ref{prop-test-set}}
\stmtproptestset
\end{qproposition}
\begin{proof}
For both $i=1,2$ and all $a \in \Sigma$ define a matrix $M_i(a) \in [0,1]^{S_i \times S_i}$ with
\[
 \big(M_i(a)\big)_{s, t} = 
 \begin{cases}
  \phi_i(s, t) & \text{if $O_i(s) = a$} \\
  0            & \text{otherwise}
 \end{cases}
 \qquad\qquad \text{for all $s, t \in S_i$.}
\]
For $i=1,2$ write $\eta_i \in \{1\}^{S_i}$ for the column vector all whose entries are~$1$. 
For $i=1,2$ and for any string $u = a_1 \ldots a_k \in \Sigma^*$ define the column vector $\eta_i(u) \in [0,1]^{S_i}$ with $\eta_i(u) = M_i(a_1) \cdot \ldots \cdot M_i(a_{k}) \cdot \eta_i$.
For all $s \in S_i$ we have, according to the definitions, the equality $\big(\eta_i(u)\big)_s = \P_{i,s}(\io{u \Sigma^{\omega}}_i)$, which is the probability that the string $u$ is output by~$H_i$ starting from~$s$. 
For a distribution $\psi_i \in \Distr{S_i}$ write $\langle \psi_i \rangle \in [0,1]^{S_i}$ for the stochastic row vector with $\langle \psi_i \rangle_s = \psi_i(s)$.
According to the definitions, we have $\pr_i(\psi_1,u) = \langle \psi_i \rangle \cdot \eta_i(u)$ for all $u \in \Sigma^*$.
Define
\[
 \eta(u) := \left(\begin{matrix} \eta_1(u) \\[1mm] - \eta_2(u)
  \end{matrix}\right) \in [0,1]^{S_1 \cup S_2}
 \qquad \text{for all $u \in \Sigma^*$.}
\]
Hence we have $\psi_1 \equiv \psi_2$ if and only if
\[
  \big(\langle \psi_1 \rangle \quad \langle \psi_2 \rangle\big) 
  \cdot 
  \eta(u) = 0 \qquad \text{for all $u \in \Sigma^*$.}
\]
It follows that we have $\psi_1 \equiv \psi_2$ if and only if
$\big(\langle \psi_1 \rangle \quad \langle \psi_2 \rangle\big)$
is orthogonal to the vector space, say~$\mathcal{V}$, spanned by
$\{\eta(u) \mid u \in \Sigma^*\}$.
Define
\[
  M(u) := \left(\begin{matrix}
M_1(u) & 0 \\
0      & M_2(u)
\end{matrix}\right)
\in [0,1]^{(S_1 \cup S_2) \times (S_1 \cup S_2)}
\]
for all $u \in \Sigma^*$.
Note that $\eta(a u) = M(a) \eta(u)$ holds for all $a \in \Sigma$ and all $u \in \Sigma^*$. 
Hence the vector space~$\mathcal{V}$ can be equivalently described as the smallest vector space that contains $\eta(\varepsilon)$ (where $\varepsilon$ denotes the empty string, i.e., all entries of $\eta(\varepsilon)$ are $\pm 1$) and satisfies $M(a) v \in \mathcal{V}$ for all $a \in \Sigma$ and all $v \in \mathcal{V}$.

We now give a polynomial-time algorithm for computing a set $\Test \subseteq \Sigma^*$.
The algorithm is as follows:
Initialize $\Test := \{\varepsilon\}$ where $\varepsilon$ denotes the empty string.
Then, as long as there are $a \in \Sigma$ and $w \in \Test$ such that $M(a) \eta(w)$ is linearly independent of $\{\eta(u) \mid u \in \Test\}$, set $\Test := \Test \cup \{a w\}$.

Now we show that the computed set $\Test$ has the properties claimed in the proposition.
Since $\mathcal{V}$ is the smallest vector space that contains $\eta(\varepsilon)$ and satisfies $M(a) v \in \mathcal{V}$ for all $a \in \Sigma$ and all $v \in \mathcal{V}$,
the set $U := \{\eta(u) \mid u \in \Test\}$ for the computed set~$\Test$ is a basis for~$\mathcal{V}$.
Since $\mathcal{V}$ is a subspace of $\mathbb{R}^{S_1 \cup S_2}$,
the dimension of~$\mathcal{V}$ is at most $m = |S_1| + |S_2|$.
Since $U$ is a basis, we have $|\Test| \le m$.
Since every string that the algorithm adds to~$\Test$ is only one letter longer than some other string already in~$\Test$, it follows that $|u| < |\Test| \le m$ holds for all $u \in \Test$.
Finally we show for all $\psi_1 \in \Distr{S_1}$ and all $\psi_2 \in \Distr{S_2}$:
\[
 \psi_1\equiv \psi_2 \quad \Longleftrightarrow\quad \forall u\in
 \Test:\ \pr_1(\psi_1,u)=\pr_2(\psi_2,u)
\]
The direction ``$\Longrightarrow$'' is immediate.
For the converse ``$\Longleftarrow$'', assume $\pr_1(\psi_1,u)=\pr_2(\psi_2,u)$ for all $u \in \Test$.
Then we have for all $u \in \Test$:
\begin{align*}
 0 & = \pr_1(\psi_1,u) - \pr_2(\psi_2,u) \\
   & = \langle \psi_1 \rangle \cdot \eta_1(u) - \langle \psi_2 \rangle \cdot \eta_2(u) \\
   & = \big( \langle \psi_1 \rangle \quad \langle \psi_2 \rangle \big) \cdot \eta(u)
\end{align*}
Since $\{\eta(u) \mid u \in \Test\} = U$ is a basis for~$\mathcal{V}$, it follows that $\big(\langle \psi_1 \rangle \quad \langle \psi_2 \rangle\big)$ is orthogonal to~$\mathcal{V}$.
We have already argued that this implies $\psi_1 \equiv \psi_2$.
This completes the proof.
\end{proof}

We prove Proposition~\ref{prop-lbound} from the main text.

\begin{qproposition}{\ref{prop-lbound}}
\stmtproplbound
\end{qproposition}
\begin{proof}
We say a state $s_1 \in S_1$ \emph{dominates} a distribution $\psi_1 \in \Distr{S_1}$ if $\psi_1(s_1) \ge \psi_1(t_1)$ holds for all $t_1 \in S_1$.
We say a pair of states $(s_1, s_2)$ is \emph{reachable} if there exists a reachable pair of distributions $(\psi_1, \psi_2)$ with $\psi_i(s_i) > 0$ for both $i=1,2$.
Note that one can compute, in polynomial time, from $H_1, H_2$ the set of all reachable pairs of states.
For $s_1 \in S_1$ define $\Unreach(s_1) := \{s_2 \in S_2 \mid (s_1, s_2) \text{ is not reachable.}\}$.
For every $s_1 \in S_1$, consider the following linear program $\mathcal{LP}(s_1)$ over a real variable~$x$ and over real variables encoding distributions $\psi_1 \in \Distr{S_1}$ and $\psi_2 \in \Distr{S_2}$:
\begin{align*}
 \text{minimize }    & x \ge 0 \\
 \text{subject to: } & \psi_1 \in \Distr{S_1} \\
                     & \psi_2 \in \Distr{S_2} \\
                     & s_1 \text{ dominates } \psi_1 \\
                     & \psi_2(s_2) = 0 \text{ for all $s_2 \in \Unreach(s_1)$ } \\
                     & -x \ \le \ \pr_1(\psi_1, u) - \pr_2(\psi_2, u) \ \le \ x  \\
                     & \qquad \text{ for all $u \in \Test$.}
\end{align*}
Note that all constraints are linear (in)equalities.
In particular, we have $\pr_i(\psi_i, u) = \sum_{s\in S_i}\psi_i(s)\cdot\P_{i,s}(\io{u\Sigma^{\omega}}_i)$.
The probabilities $\P_{i,s}(\io{u\Sigma^{\omega}}_i)$ can be computed in polynomial time.
(Those probabilities are computed already when computing the set~$\Test$ according to the proof of Proposition~\ref{prop-test-set}:
they are the probabilities in the vectors $\eta_i(u)$ defined there.)

For every $s_1 \in S_1$, let $c(s_1)$ denote the optimum solution (minimizing~$x$) of~$\mathcal{LP}(s_1)$.
Define $c := \min \{c(s_1) \mid s_1 \in S_1\}$.
Note that $c$ can be computed in polynomial time.
We show that $c$ has the properties claimed by the proposition.

First we show that $\dist(\psi_1, \psi_2) \ge c$ holds for all reachable pairs $(\psi_1, \psi_2)$.
Towards a contradiction suppose that there is a reachable pair $(\psi_1, \psi_2)$ with $\dist(\psi_1, \psi_2) < c$.
Let $s_1 \in S_1$ be a state that dominates~$\psi_1$.
Since $(\psi_1, \psi_2)$ is reachable, we have $\psi_2(s_2) = 0$ for all $s_2 \in \Unreach(s_1)$.
By the definition of $\dist(\psi_1, \psi_2)$, we have
\[
 -\dist(\psi_1, \psi_2) \ \le \ \pr_1(\psi_1, u) - \pr_2(\psi_2, u) \ \le \ \dist(\psi_1, \psi_2)   
\]
for all $u \in \Test$.
It follows that $x := \dist(\psi_1, \psi_2)$ along with $\psi_1, \psi_2$ is a feasible solution of the linear program $\mathcal{LP}(s_1)$.
Since $c(s_1)$ is optimal, we have $c(s_1) \le \dist(\psi_1, \psi_2)$.
By our assumption we have $\dist(\psi_1, \psi_2) < c$, hence $c(s_1) < c$.
But by the definition of~$c$ we have $c \le c(s_1)$, a contradiction.
We conclude that $\dist(\psi_1, \psi_2) \ge c$ holds for all reachable pairs $(\psi_1, \psi_2)$.

Finally, we show $c>0$.
Towards a contradiction suppose $c=0$.
So by definition of~$c$ there is $s_1 \in S_1$ with $c(s_1) = 0$.
Thus, $\mathcal{LP}(s_1)$ has a solution with $x=0$.
That is, there exist $\psi_1 \in \Distr{S_1}$ and $\psi_2 \in \Distr{S_2}$ such that $s_1$ dominates~$\psi_1$,
and
\begin{equation}
\psi_2(s_2) = 0 \ \text{ holds for all $s_2 \in \Unreach(s_1)$,} \label{eq-lbound-proof-1}
\end{equation}
and $\pr_1(\psi_1, u) = \pr_2(\psi_2, u)$ holds for all $u \in \Test$.
By Proposition~\ref{prop-test-set}, the last fact implies
\begin{equation}
\psi_1 \equiv \psi_2\,. \label{eq-lbound-proof-2}
\end{equation}
Since $s_1$ dominates~$\psi_1$, we have
\begin{equation}
\psi_1(s_1) > 0\,. \label{eq-lbound-proof-3}
\end{equation}
It follows directly from~\cite[Theorem~21]{14CK-LICS} that \eqref{eq-lbound-proof-1}--\eqref{eq-lbound-proof-3}
together imply that we have $d(H_1, H_2) < 1$ for the the total variation distance~$d$ defined in the beginning of Section~\ref{sec-distinguish}.
But then Proposition~\ref{prop-reduction-to-dist-one} implies that $H_1, H_2$ are not distinguishable, which is a contradiction.
Hence $c>0$ must hold.
This concludes the proof.
\end{proof}

We prove Proposition~\ref{prop-expensive-c} from the main text.

\begin{qproposition}{\ref{prop-expensive-c}}
\stmtpropexpensive
\end{qproposition}
\begin{proof}
The reachable pairs $(S_1', S_2') \in 2^{S_1} \times 2^{S_2}$ can be computed in exponential time.
So it suffices to show that one can compute, for a fixed reachable pair $(S_1', S_2') \in 2^{S_1} \times 2^{S_2}$, the value
\begin{align*}
 c_{S_1', S_2'} &:=
          \min_{\psi_1 \in \Distr{S_1'}} \
          \min_{\psi_2 \in \Distr{S_2'}} \\
 &  \qquad\qquad \max_{U \subseteq \Sigma^m} \
          \big( \pr_1(\psi_1, U) - \pr_2(\psi_2, U) \big)
\end{align*}
in exponential time.
Consider the following linear program, similar to the one from the proof of Proposition~\ref{prop-lbound},
with variables $x_u$ for $u \in \Sigma^m$ and variables encoding distributions $\psi_1, \psi_2$:
\begin{align*}
 \text{minimize }    & \sum_{u \in \Sigma^m} x_u \\
 \text{subject to: } & \psi_1 \in \Distr{S_1} \\
                     & \psi_2 \in \Distr{S_2} \\
                     & 0 \le x_u && \text{for all $u \in \Sigma^m$} \\
                     & \pr_1(\psi_1, u) - \pr_2(\psi_2, u) \ \le \ x_u && \text{for all $u \in \Sigma^m$}
\end{align*}
This linear program has exponential size.
We show that its optimal solution is $c_{S_1', S_2'}$.

First we show that it has a feasible solution whose value is $c_{S_1', S_2'}$.
Let $\psi_1, \psi_2$ be the distributions that attain the minimum from the definition of~$c_{S_1', S_2'}$.
Let $U$ be a set that attains the maximum from the definition of~$c_{S_1', S_2'}$.
We can take $U = \{u \in \Sigma^m \mid \pr_1(\psi_1, u) \ge \pr_2(\psi_2, u)\}$.
Let $x_u = \pr_1(\psi_1, u) - \pr_2(\psi_2, u)$ for all $u \in U$, and let $x_u = 0 $ for all $u \in \Sigma^m - U$.
Then the solution with those $x_u$ and with $\psi_1, \psi_2$ is feasible.
Moreover, its value is:
\begin{align*}
\sum_{u \in \Sigma^m} x_u
& \ = \ \sum_{u \in U} x_u \\
& \ = \ \sum_{u \in U} \left( \pr_1(\psi_1, u) - \pr_2(\psi_2, u) \right) \\
& \ = \ \pr_1(\psi_1, U) - \pr_2(\psi_2, U) \\
& \ = \ c_{S_1', S_2'}
\end{align*}

For the converse, we show that $c_{S_1', S_2'}$ is a lower bound to the value of any feasible solution.
Let $(x_u)_{u \in \Sigma^m}$ along with $\psi_1, \psi_2$ denote a feasible solution.
Let $U$ be a set that attains the maximum in $\max_{U \subseteq \Sigma^m} \big( \pr_1(\psi_1, U) - \pr_2(\psi_2, U) \big)$.
We can take $U = \{u \in \Sigma^m \mid \pr_1(\psi_1, u) \ge \pr_2(\psi_2, u)\}$.
Hence we have:
\begin{align*}
& \ \sum_{u \in \Sigma^m} x_u \\
& \ \ge \ \sum_{u \in U} x_u \\
   &\qquad\qquad \left(\text{$x_u \ge 0$ from the linear program}\right) \\
& \ \ge \ \sum_{u \in U} \left( \pr_1(\psi_1, u) - \pr_2(\psi_2, u) \right) \\
   &\qquad\qquad \left(\text{$x \ge \pr_1(\psi_1, u) - \pr_2(\psi_2, u)$ from the lin.~program}\right) \\
& \ = \ \pr_1(\psi_1, U) - \pr_2(\psi_2, U) \\
& \ = \ \max_{U \subseteq \Sigma^m} \
          \big( \pr_1(\psi_1, U) - \pr_2(\psi_2, U) \big) \\
& \ \ge \ c_{S_1', S_2'} \\
   &\qquad\qquad \left(\text{definition of~$c_{S_1', S_2'}$}\right)
\end{align*}
We conclude that $c_{S_1', S_2'}$ is an optimal solution of the linear program.
\end{proof}

%% file: app-verification.tex
\section{Proofs of Section~6} \label{app-verification}

We prove Proposition~\ref{prop-conditioned} from the main text.

\begin{qproposition}{\ref{prop-conditioned}}
\stmtpropconditioned
\end{qproposition}

\begin{proof}
By symmetry, it suffices to provide the construction for~$H_1$.
Let $H = (G, O, s_0, \Class)$ be the given cHMC
 with $G = (S, R, \phi)$ a Markov chain.
Define
\[
 S_1 := \{s \in S \mid \P_s(\Bad_s) > 0\}\,.
\]
Note that $s_0 \in S_1$.
Define $G_1 := (S_1, R_1, \phi_1)$ with
 $R_1 := R \cap (S_1 \times S_1)$  and
 \[
  \phi_1(s,t) := \frac{\phi(s,t) \cdot \P_t(\Bad_t)}{\P_s(\Bad_s)}
  \qquad \text{for all $(s,t) \in R_1$}.
 \]
Finally, take $H_1 := (G_1, O_1, s_0)$ where $O_1$ equals~$O$ restricted to~$S_1$.

We show that the measures $\P_1(\cdot)$ and~$\P(\cdot \mid \Bad)$ are equal.
By definition, it suffices to show that
 they are equal on the cylinder sets $\{s_0 r\} S_1^\omega$ 
 for all $r \in S_1^*$.
We show by induction on the length of~$r$ that 
\[
 \P_{1,s}(\{s r\} S_1^\omega) \cdot \P_s(\Bad_s) 
 = \P_s(\{s r\} S_1^\omega \cap \Bad_s)
 \quad \text{$\forall\,s \in S_1$.}
\]
For the induction base, let $r$ be empty.
Then the claim follows from $\P_{1,s}(\{s\} S_1^\omega) = 1$
 and $\Bad_s \subseteq \{s\} S_1^\omega$.
For the induction step, let $t \in S_1$ and $r \in S_1^*$.
We want to show:
\begin{equation} \label{eq-condition-inductive}
 \P_{1,s}(\{s t r\} S_1^\omega) \cdot \P_s(\Bad_s) 
 = \P_s(\{s t r\} S_1^\omega \cap \Bad_s)
\end{equation}
If $(s,t) \not\in R_1$ then both sides of~\eqref{eq-condition-inductive} are zero.
So let $(s,t) \in R_1$.
Then we have:
\begin{align*}
& \P_{1,s}(\{s t r\} S_1^\omega) \cdot \P_s(\Bad_s) \\
& = \phi_1(s,t) \cdot \P_{1,t}(\{t r\} S_1^\omega) \cdot \P_s(\Bad_s) \\
& = \frac{\phi(s,t) \cdot \P_t(\Bad_t)}{\P_s(\Bad_s)} \cdot \P_{1,t}(\{t r\} S_1^\omega) \cdot \P_s(\Bad_s) \\
& = \phi(s,t) \cdot \P_t(\{t r\} S_1^\omega \cap \Bad_t) \qquad\qquad \text{by the ind.~hyp.} \\
& = \P_s(\{s t r\} S_1^\omega \cap \Bad_s)
\end{align*}
This shows~\eqref{eq-condition-inductive} and hence the proposition.
\end{proof}